\documentclass[12pt, letterpaper]{article}
\usepackage[margin=1.25in]{geometry}
\usepackage{multirow, multicol}
\usepackage{booktabs}
\usepackage{dsfont,amssymb,amsmath,amsthm,centernot,graphicx,accents}
\usepackage{enumitem}
\usepackage{natbib}
\usepackage[colorlinks = true, citecolor = blue]{hyperref}
\usepackage{cleveref}
\usepackage{setspace}
\usepackage{comment}
\usepackage{algorithm}
\usepackage{algpseudocode}
\usepackage{ragged2e}
\usepackage{caption}
\usepackage{subcaption}
\usepackage{float}

\crefname{appendix}{appendix}{appendices}
\Crefname{appendix}{Appendix}{Appendices}

\usepackage{pgfplots} 
\pgfplotsset{compat = newest} 
\usetikzlibrary{positioning, arrows.meta} 
\usepgfplotslibrary{fillbetween}

\newtheorem{theorem}{Theorem}
\newtheorem*{theorem*}{Theorem}
\Crefname{theorem}{Theorem}{Theorems}
\crefname{theorem}{theorem}{theorems}

\newtheorem{proposition}{Proposition}
\newtheorem*{proposition*}{Proposition}
\Crefname{proposition}{Proposition}{Propositions}
\crefname{proposition}{proposition}{propositions}

\newtheorem{corollary}{Corollary}
\newtheorem*{corollary*}{Corollary}
\Crefname{corollary}{Corollary}{Corollaries}
\crefname{corollary}{corollary}{corollaries}

\newtheorem*{claim*}{Claim}
\Crefname{claim}{Claim}{Claims}
\crefname{claim}{claim}{claims}

\newtheorem{lemma}{Lemma}
\newtheorem*{lemma*}{Lemma}
\Crefname{lemma}{Lemma}{Lemmas}
\crefname{lemma}{lemma}{lemmas}

\newtheorem*{conjecture*}{Conjecture}
\Crefname{conjecture}{Conjecture}{Conjectures}
\crefname{conjecture}{conjecture}{conjectures}

\theoremstyle{definition}

\newtheorem{definition}{Definition}
\newtheorem*{definition*}{Definition}
\Crefname{definition}{Definition}{Definitions}
\crefname{definition}{definition}{definitions}

\newtheorem{assumption}{Assumption}
\newtheorem*{assumption*}{Assumption}
\Crefname{assumption}{Assumption}{Assumptions}
\crefname{assumption}{assumption}{assumptions}

\newtheorem{remark}{Remark}
\newtheorem*{remark*}{Remark}
\Crefname{remark}{Remark}{Remarks}
\crefname{remark}{remark}{remarks}

\newtheorem{example}{Example}
\newtheorem*{example*}{Example}
\Crefname{example}{Example}{Examples}
\crefname{example}{example}{examples}

\newcommand{\E}{\mathbb{E}}
\newcommand{\R}{\mathds{R}}

\newcommand{\Var}{\textup{Var}}
\newcommand{\diag}{\textup{diag}}

\newcommand{\tr}{\operatorname{tr}}
\newcommand{\spann}{\operatorname{span}}

\newcommand{\HAC}{\text{HAC}}
\newcommand{\clu}{\text{clu}}

\newcommand{\Sym}{\mathrm{Sym}}
\newcommand{\range}{\mathrm{range}}

\newcommand{\Cov}{\mathrm{Cov}}

\allowdisplaybreaks

\title{Variance Estimation with Dependence and Heterogeneous Means}

\date{\today }
\author{Luther Yap\thanks{Email: \texttt{lyap@nus.edu.sg}. Department of Economics, National University of Singapore. The latest version is available at \url{https://lutheryap.github.io/files/twclus_het_means_wp.pdf}}}

\begin{document}


\maketitle

\begin{abstract}

This paper develops a framework for variance estimation under dependence and heterogeneous means. 
This paper shows that consistent estimation of the variance target is impossible in general, and  characterizes necessary and sufficient conditions for conservative variance estimation using dual cones. 
To choose among the valid estimators, this paper formulates three criteria --- minimal correction, pointwise level estimand, and pointwise MSE --- and shows how an eigenvalue truncation solution is optimal under all three criteria.
This characterization and solution allow us to assess if existing variance estimators are valid and optimal in their respective settings, and construct the first optimal variance estimator that is simultaneously robust to heterogeneous means and cross-cluster serial correlation.
\end{abstract}


\vskip 0.2in


\section{Introduction}

Suppose $Y_n \in \mathbb{R}^n$ is a vector of random variables with $E[Y_n] = \mu \in \mathcal{M}_n$ and $Var(Y_n) = \Sigma \in \mathcal{S}_n$. 
We are interested in a variance target $\theta_n = \tr(W_n \Sigma)$ for some known weight matrix $W_n$.
The object $\theta_n$ can be interpreted as the oracle variance for $\sum_i Y_{n,i}$, which is used for inference after applying a central limit theorem.
As an example, with independent observations, $W_n = I_n$, but $W_n$ also encodes dependent observations such as one/ two- way clustering and serial correlation.
To obtain the variance target, we rely on a variance estimator that takes a quadratic form $Y_n^\prime A_n Y_n$, for some $A_n$. 
Observe that  $E[Y_n' A_n Y_n] = \tr(A_n \Sigma) + \mu' A_n \mu$. 
When the mean is known or consistently estimable, the second term $\mu' A_n \mu$ is innocuous and can be removed by demeaning, and the first term can recover the target. 
The difficulty confronted in this paper is when means can be arbitrarily heterogeneous, so the means are no longer consistently estimable.
Considering heterogeneous means with a general dependence structure is the novel problem addressed in this paper.

Heterogeneous means arise in several settings. 
Under design-based uncertainty, where potential outcomes are fixed and cluster dependence reflects correlated assignment or sampling \citep{abadie2023should}, the score for each observation is heterogeneous and not consistently estimable by construction. 
In asset pricing, when we condition on the set of assets and time observed, any omitted priced factor leaves the regression score with a mean that varies across assets and time. 
In both cases the individual means are not consistently estimable even though the grand mean is, which is the regime in which a central limit theorem governs the sum while the summands' means remain unidentified and the variance of the sum becomes difficult to estimate.

This paper makes three contributions.

First, this paper establishes the limits of identification in the problem in \Cref{sec:theory}. 
With dependence and arbitrary mean heterogeneity, this paper shows that consistent estimation of the variance target $\theta_n = \tr(W_n\Sigma)$ is impossible in general (\Cref{thm:impossibility}), because the second moment alone is uninformative of the split between the covariance component and the mean component. 
The obstruction is one of identification, not of finite samples. 
Since consistency is unattainable, we can only pursue the second best goal, a conservative variance estimand, one that is guaranteed not to fall below the target, so that the associated statistical test remains valid (i.e., with size guarantee). 
Seemingly disparate estimators are quadratic forms governed by the same two objects: a linear subspace $\mathcal M_n$ of admissible mean vectors and a closed convex cone $\mathcal S_n$ of admissible covariance matrices. 
Necessary and sufficient conditions for conservativeness in terms of these spaces are derived (\Cref{thm:conservative}): $Y_n'A_nY_n$ is uniformly conservative for $\theta_n$ if and only if $A_n-W_n$ lies in the dual cone $\mathcal S_n^*$ and $A_n$ is positive semidefinite on $\mathcal M_n$. 

In light of how there are many variance estimators that are valid, the second contribution of the paper is to propose an optimal variance estimator in \Cref{sec:optimal-conservative}. 
Optimality is defined relative to criteria, so this paper proposes three natural criteria --- minimal correction (which respects the dependence structure), minimum level of the variance estimand (in the interest of having the smallest variance used in the t-statistic and hence the most powerful test), and mean-squared error (MSE) (which is a common benchmark that assesses the quality of an estimator). 
The solution proposed in this paper is to set $A_n$ to be an eigenvalue truncation (ET) of $W_n$.
Under some conditions, this solution is optimal under all three criteria. (\Cref{thm:eig-truncation})
If these conditions do not hold, then these criteria can nonetheless be transformed into constrained conic optimization problems, so an optimal solution can be obtained numerically. (\Cref{prop:conicity})

Given the framework and conditions for validity and optimality, the third contribution of the paper is to assess if existing variance estimators in economically relevant settings with dependence are valid and optimal under heterogeneity, and if they are not, obtain the optimal conservative variance estimator.
\Cref{sec:examples-and-optimal} considers several existing examples. 
The plug-in one-way cluster robust variance estimator is valid under mean heterogeneity (a result known from \citet{abadie2023should}), but now we also certify that it is optimal. 
In two-way clustering, the plug-in variance estimator from \citet{cameron2011robust} (CGM) is not valid under mean heterogeneity, while the CGM2 modification from \citet{davezies2021empirical} is valid --- these results are known from \citet{xu2024clustering}, and are corollaries of our setting. 
While CGM2 is valid, the new result is that CGM2 is not optimal as it does not match the ET solution --- the difference is that CGM2 adds observation-specific variances to CGM, while ET adds the observation-specific within-transformed variances, which results in a smaller variance estimand.
\Cref{sec:application_panel_dep} extends the two-way cluster dependence to allow for cross-cluster weak dependence, a dependence structure considered in \citet{chiang2024standard} (CHS).
The plug-in variance estimator from CHS is invalid under mean heterogeneity, so a novel optimal conservative variance estimator is proposed as a corollary of this paper's theory. 
This section also develops the $\psi$-dependence framework, the central limit theorem, and consistency.
Simulations and an empirical application close the section.

\textbf{Related literature.}
This paper relates to three strands of literature. 

The first is variance estimation in dependent arrays. 
\citet{cameron2011robust} and \citet{chiang2024standard} provide variance estimators under the two-way clustering and two-way clustering with cross cluster correlations, but both assume mean homogeneity. 
This paper extends their variance estimators to be robust to mean heterogeneity. 
Beyond variance estimation in the setting with two-way clustering with cross cluster correlations, this paper also provides a central limit theorem that is weaker than the one used in CHS, as this paper relies on the $\psi$-dependence limit theory from \citet{kojevnikov2021limit} (KMS) that does not rely on an Aldous-Hoover representation, thereby accommodating data-generating processes outside the representation, similar to \citet{yap2025asymptotic}.

The second is distributional and mean heterogeneity. 
In time series, \citet{chan2022optimal} and \citet{casini2023theory} permit time-varying means and second moments but require enough regularity to estimate the mean function consistently, whereas we impose no regularity on the mean sequence and instead modify the estimand to guarantee conservativeness. 
The interaction between dependence and mean heterogeneity has been rather disparate in design-based settings: settings with independence \citep{abadie2020sampling} or one-way clustering \citep{abadie2023should} have conservative plug-in variances, while two-way clustering does not \citep{xu2024clustering}. 
This paper's framework unifies these dependence structures, and the conservativeness criterion gives a general and practical way to check if a dependence-robust plug-in is conservative. 
Further, this paper extends conservative variance estimation under mean heterogeneity beyond these settings to the CHS dependence structure, and additionally derives optimal variances.  

Third, a long tradition in time-series and clustered variance estimation is concerned that a quadratic form $Y_n'W_nY_n$ can return a negative number in finite samples when $W_n$ is not PSD. 
\citet{newey1987asimple} address this ex ante by adopting the Bartlett kernel; \citet{andrews1991heteroskedasticity} characterizes the class of kernels delivering positive-semidefinite estimates and identifies the quadratic-spectral kernel as minimizing asymptotic mean-squared error within it. 
A complementary, ex post tradition takes a possibly non-PSD estimate and projects it onto the PSD cone by eigenvalue correction, the operation which appears in clustered inference as the eigenvalue correction of \citet{chiang2024standard}. 
Our framework differs in two respects. 
First, the non-PSD-ness we confront is not a sampling-noise artifact of a particular realization but a structural feature of the population weight matrix that the dependence design dictates so it does not vanish as the sample grows and cannot be diagnosed from any single estimate. 
Second, and consequently, our optimality is stated for the estimand, not the estimate: where the classical corrections enforce a threshold on a realization or minimize the distance of an estimate to the positive-semidefinite cone, ours minimizes the mean-induced bias of the targeted variance subject to conservativeness.

\section{Framework and Quadratic Representation}
\label{sec:framework}

\subsection{Statistical model and parameter}

For each $n\geq 1$, let $Y_n=(Y_{n,1},\dots,Y_{n,n})'\in\mathbb{R}^n$ be
observed, distributed according to $P_n$. Define
\[
\mu(P_n):=E_{P_n}[Y_n],\qquad
\Sigma(P_n):=\Var_{P_n}(Y_n).
\]
This paper will drop the indexing on $P_n$ when it is without ambiguity. 

The researcher commits to two structural objects:
\begin{itemize}
  \item a linear subspace $\mathcal{M}_n\subseteq\mathbb{R}^n$ of admissible
        mean vectors;
  \item a closed convex cone $\mathcal{S}_n\subseteq\mathrm{Sym}_n^+$ of
        admissible covariance matrices, where $\mathrm{Sym}_n^+$ denotes the
        positive semidefinite cone in the space of symmetric $n\times n$
        matrices $\mathrm{Sym}_n$.
\end{itemize}
The model is
\[
\mathcal{P}_n(\mathcal{M}_n,\mathcal{S}_n)
:=
\bigl\{P_n: E_{P_n}\|Y_n\|^2<\infty,\
\mu(P_n)\in\mathcal{M}_n,\ \Sigma(P_n)\in\mathcal{S}_n\bigr\}.
\]
Let $W_n\in\mathrm{Sym}_n$ be a known symmetric weight matrix. The parameter of
interest is
\[
\theta_n(P_n):=\tr\!\bigl(W_n\Sigma(P_n)\bigr),
\qquad P_n\in\mathcal{P}_n(\mathcal{M}_n,\mathcal{S}_n).
\]

An important context where $\theta_n$ is interesting is whenever some central limit theorem is applied while allowing for dependence across observations and mean heterogeneity. The central limit theorem yields $\left( \sum_{i=1}^n Y_{n,i} - E[Y_{n,i}] \right) / \sqrt{\Var \left(  \sum_{i=1}^n Y_{n,i} \right)} \xrightarrow{d} N(0,1)$. While $\frac{1}{n} \sum_{i=1}^n E[Y_{n,i}]$ is consistently estimable from the theorem, the individual means may not be when there is sufficient heterogeneity, which makes the estimation of $\Var \left(  \sum_{i=1}^n Y_{n,i} \right)$ tricky when conducting statistical inference. 
When $W_n$ corresponds to an adjacency matrix and encodes several dependence structures, this parameter of interest $\theta_n$ converges to $\Var \left(  \sum_{i=1}^n Y_{n,i} \right)$ under some conditions, even when dependence across $Y_{n,i}$ is permitted. 
For instance, $W_n= I_n$ corresponds to independent observations. In time series settings with a kernel adjustment for the variance, $W_n$ would reflect these kernel adjustments.

The pair $(\mathcal{M}_n,\mathcal{S}_n)$ encodes the researcher's structural commitments. 
Larger $\mathcal{M}_n$ and $\mathcal{S}_n$ correspond to weaker assumptions and stronger demands for robustness. 
Two extreme cases are instructive. 
At one end, $\mathcal{S}_n=\mathrm{Sym}_n^+$ imposes no structure on the covariance and yields the maximally dependence-robust posture: the researcher seeks conservativeness against any covariance compatible with $\Sigma\succeq 0$. 
At the other end, $\mathcal{S}_n=\{\sigma^2I_n:\sigma^2\geq0\}$ encodes iid sampling and admits the classical demeaned sample variance. 
Intermediate cones encode block independence (clusters), bandedness ($m$-dependence), or diagonal structure (without correlation), each yielding its own optimal estimator.

\textbf{Examples of $\mathcal{S}_n$.}
\begin{itemize}
\item $\mathcal{S}_n^{\mathrm{iid}}:=\{\sigma^2I_n:\sigma^2\geq0\}$:
      iid sampling.
\item $\mathcal{S}_n^{\mathrm{diag}}:=\{\Sigma\succeq0:\Sigma\text{ diagonal}\}$:
      heterogeneous variances but uncorrelated.
\item $\mathcal{S}_n^{\mathrm{block}}:=\{\Sigma\succeq0:\Sigma_{ij}=0\text{ if }g(i)\neq g(j)\}$:
      cluster independence.
\item $\mathcal{S}_n^{\mathrm{band}(b)}:=\{\Sigma\succeq0:\Sigma_{ij}=0\text{ if }|i-j|>b\}$:
      $b$-dependence.
\end{itemize}

Similarly, $\mathcal{M}_n$ can encode the known structure of the problem. Fully unrestricted means corresponds to $\mathcal{M}_n = \mathbb{R}^n$, while constant means corresponds to $\mathcal{M}_n = \mathrm{span} \{1_n\}$, as two extreme cases. If we have heterogeneous means but are able to consistently estimate the mean of all elements in $\mu$, then after demeaning by the grand mean, we may restrict our attention to $\mathcal{M}_n =  \{ \mu: 1_n^\prime \mu =0 \} = 1_n^\perp$.

\subsection{Quadratic-form estimators and dual cones}

A quadratic-form estimator takes the form $\widehat{V}_{A,n}(Y_n):=Y_n'A_nY_n$ for some $A_n\in\mathrm{Sym}_n$, with expectation
\begin{equation}
V_{A,n}(P_n)
:=
E_{P_n}[Y_n'A_nY_n]
=
\tr\!\bigl(A_n\Sigma(P_n)\bigr)+\mu(P_n)'A_n\mu(P_n).
\label{eq:plug-decomp}
\end{equation}
$V_{A,n}(P_n)$ is an estimand, which is not necessarily the target $\theta_n(P_n)$. 

The canonical plug-in estimator for the sum of $Y_n$ (i.e., $1_n^\prime Y_n$) corresponds to $A_n=W_n$, with bias
\[
V_{W,n}(P_n)-\theta_n(P_n)=\mu(P_n)'W_n\mu(P_n).
\]
This bias can be negative when $W_n$ has negative directions inside $\mathcal{M}_n$, arising in anticonservativeness in inference.

To analyze conservativeness across the full pair $(\mathcal{M}_n,\mathcal{S}_n)$, we use cone duality. 
The \emph{dual cone} of $\mathcal{S}_n$ is
\begin{equation}
\mathcal{S}_n^*
:=
\bigl\{B\in\mathrm{Sym}_n:\tr(B\Sigma)\geq0\text{ for all }\Sigma\in\mathcal{S}_n\bigr\}.
\label{eq:dual-cone}
\end{equation}
Since $\mathcal{S}_n\subseteq\mathrm{Sym}_n^+$ and $\mathrm{Sym}_n^+$ is self-dual, $\mathcal{S}_n^*\supseteq\mathrm{Sym}_n^+$, with equality when $\mathcal{S}_n=\mathrm{Sym}_n^+$. 
For strict subcones, $\mathcal{S}_n^*$ contains symmetric matrices that are not positive semidefinite, and this larger admissible region is precisely what licenses non-PSD corrections such as demeaning.

Examples of dual cones, all easily verified from \eqref{eq:dual-cone}:
\[
(\mathcal{S}_n^{\mathrm{iid}})^*=\{A:\tr(A)\geq0\},\quad
(\mathcal{S}_n^{\mathrm{diag}})^*=\{A:\diag(A)\geq0\}.
\]
For $\mathcal{S}_n^{\mathrm{block}}$, $A\in\mathcal{S}_n^{\mathrm{block},*}$ if and only if each within-cluster diagonal block of $A$ is positive semidefinite, with arbitrary off-block entries.
The covariance cones sit in a partial order: $\mathcal{S}_n^{\mathrm{iid}}\subset\mathcal{S}_n^{\mathrm{diag}}\subset \mathcal{S}_n^{\mathrm{block}}\subset \mathrm{Sym}_n^+$.
Dual cones reverse the order.

\section{Impossibility and Conservativeness}
\label{sec:theory}

\subsection{When consistent estimation is impossible}

The same second moment matrix may admit different decompositions into a mean component and a covariance component:
\[
E[Y_nY_n']
=
\Sigma(P_n)+\mu(P_n)\mu(P_n)'.
\]
The target \(\theta_n(P_n)=\tr(W_n\Sigma(P_n))\), however, depends only on the covariance component.
Hence, if the model permits a rank-one term \(\mu\mu'\) to be shifted between covariance and mean while preserving admissibility, then the target changes but the expectation of every quadratic estimator does not. 
This yields an impossibility result for uniform consistency over the quadratic class.

\begin{theorem}[Impossibility of consistency]
\label{thm:impossibility}
Fix $n$ and consider the model $\mathcal{P}_n(\mathcal{M}_n, \mathcal{S}_n)$.
Suppose there exist $\Sigma_0 \in \mathcal{S}_n$ and $v_n \in \mathcal{M}_n$ with $\Sigma_0 + v_n v_n' \in \mathcal{S}_n$. 
Then, for any $A_n \in \mathrm{Sym}_n$,
\[
\sup_{P_n \in \mathcal{P}_n(\mathcal{M}_n, \mathcal{S}_n)} \big|\mathbb{E}_{P_n}[\widehat{V}_{A,n}] - \theta_n(P_n)\big| \geq \frac{1}{2} |v_n' W_n v_n|.
\]
\end{theorem}

To gain some intuition for the result, suppose $\mathcal{S}_n$ is the set of diagonal matrices, and $v_n = (1,0,\cdots, 0)^\prime \in \mathcal{M}_n$. 
Then, the second moment alone cannot distinguish $\mu_1 =1$ and $\Sigma_{11} =1$ from $\mu_1 =0$ and $\Sigma_{11} =2$. 

As a corollary, write $\lambda_n \asymp \theta_n$ for the scale of the target. 
If $v_n'W_nv_n \ne o(\lambda_n)$, then the relative worst-case bias of every quadratic estimator is bounded below,
\[
\frac{1}{\lambda_n}\sup_{P_n\in\mathcal P_n(\mathcal M_n,\mathcal S_n)}
\big|\E_{P_n}[\widehat V_{A,n}]-\theta_n(P_n)\big|\ \ge\ \frac{|v_n'W_nv_n|}{2\lambda_n}
\ \not\to\ 0 .
\]
Consequently no quadratic estimator is consistent in that $\widehat V_{A,n}/\theta_n\xrightarrow{p}1$, uniformly over the model: at the two Gaussian points $P_n,Q_n$ of \Cref{thm:impossibility} the normalized forms $\widehat V_{A,n}/\theta_n$ are $L^2$-bounded along the sequence, hence uniformly integrable, so consistency would force relative asymptotic unbiasedness $\E_{P_n}[\widehat V_{A,n}]/\theta_n\to1$ at each point, contradicting the display.

The proof rests on a two-point moment-matching argument: two distributions $P_n, Q_n \in \mathcal{P}_n$ can have the same expectation of every quadratic functional yet different values of $\theta_n$, rendering them indistinguishable in expectation to the quadratic class. 
The factor reflects that any estimator's expectation is a single number that cannot simultaneously match two admissible distributions with the same quadratic-form expectation but targets differing by $v_n^\prime W_n v_n$, so its bias against the worse of the two is at least half that gap.

\begin{example}[Demeaned sample variance]
\label{ex:demean}
The following example does not satisfy the conditions of the impossibility result. 
Let $W_n=I_n$, $\mathcal{M}_n=\spann(\mathbf{1}_n)$, and $\mathcal{S}_n=\mathcal{S}_n^{\mathrm{iid}}$, so $\theta_n=n\sigma^2$. 
Then $\mathcal{S}_n^*=\{A:\tr(A)\geq0\}$, and
\[
A_n=\frac{n}{n-1}M_{\mathbf{1}_n},\qquad
M_{\mathbf{1}_n}:=I_n-\frac{1}{n}\mathbf{1}_n\mathbf{1}_n'.
\]
The estimator $Y_n'A_n Y_n=\frac{n}{n-1}\sum_i(Y_{n,i}-\bar Y_n)^2$ is exactly unbiased for $n\sigma^2$ on the model.
 
The impossibility in \Cref{thm:impossibility} requires the cone $\mathcal{S}_n$ to admit a perturbation along the mean direction: there must exist $\Sigma_0 \in \mathcal{S}_n$ and $v_n \in \mathcal{M}_n$ with $\Sigma_0 + v_n v_n' \in \mathcal{S}_n$ and $v_n' W_n v_n \neq 0$. 
In the iid case this feasibility condition fails.
Any $v_n \in \mathcal{M}_n = \mathrm{span}\{1_n\}$ has the form $v_n = \alpha 1_n$, so $v_n v_n' = \alpha^2 1_n 1_n'$, a rank-one matrix with nonzero off-diagonal entries. 
For any $\Sigma_0 = \sigma_0^2 I_n \in \mathcal{S}_n^{\mathrm{iid}}$, the perturbed covariance $\Sigma_0 + \alpha^2 1_n 1_n'$ is not a scalar multiple of $I_n$ and therefore lies outside $\mathcal{S}_n^{\mathrm{iid}}$ whenever $\alpha \neq 0$.
The cone is too narrow for the adversary to hide mean variation inside a covariance perturbation, and consistent estimation of $\theta_n = n\sigma^2$ is restored.
\end{example}

\subsection{Characterization of conservative estimands}

In light of how consistent variance estimation may be impossible, we may use conservative variance estimators to avoid being oversized in inference. 
The following result provides necessary and sufficient conditions for a variance estimand to be conservative using the dual cone representation.

\begin{theorem}[Conservativeness via dual cones]
\label{thm:conservative}
The estimand $V_{A,n}(P_n)=E_{P_n}[Y_n'A_nY_n]$ satisfies $V_{A,n}(P_n)\geq
\theta_n(P_n)$ for every $P_n\in\mathcal{P}_n(\mathcal{M}_n,\mathcal{S}_n)$
if and only if
\begin{align}
\text{(C1)}\quad & A_n-W_n\in\mathcal{S}_n^*, \label{eq:C1}\\
\text{(C2)}\quad & \mu'A_n\mu\geq0\text{ for all }\mu\in\mathcal{M}_n.
\label{eq:C2}
\end{align}
\end{theorem}

When $\mathcal{S}_n=\mathrm{Sym}_n^+$, self-duality gives $\mathcal{S}_n^*=\mathrm{Sym}_n^+$ and (C1) reduces to $A_n-W_n\succeq0$ on $\mathbb{R}^n$, since $\tr(B\Sigma)\geq0$ for every PSD $\Sigma$ if and only if $B$ is itself PSD. 
For strict subcones, (C1) is weaker, admitting estimators that exploit covariance structure to sharpen the correction.
A corollary of this lemma is that the plug-in variance estimator with $W_n$ is conservative whenever $W_n$ is positive semidefinite.

However, if $A_n = W_n$ is not positive semidefinite, (C1) can still be satisfied, and (C2) is still satisfied if $\mathcal{M}_n$ is sufficiently restricted (e.g., $\mathcal{M}_n = \{ 0 \}$). 
Nonetheless, if $\mathcal{M}_n = \mathbb{R}^n$, then (C2) is equivalent to positive semidefiniteness of $A_n= W_n$, so $W_n$ not being positive semidefinite implies that $V_{A,n}(P_n) < \theta_n (P_n)$ for some $\mu \in \mathcal{M}_n$. This result will be evident in CGM and $m$-dependence.

The characterization of conservativeness appears similar to the classical positive-semidefiniteness corrections in the variance-estimation literature. 
\Cref{prop:estimand-vs-estimate} states the connection between the estimator and the estimand: if $W_n$ is positive semidefinite, then both the estimator is always positive semidefinite, and the estimand will be valid.

\begin{proposition}[Correction of Estimand vs Estimate]
\label{prop:estimand-vs-estimate}
Consider the plug-in $\widehat V_{W,n}=Y_n'W_nY_n$ with target
$\theta_n(P_n)=\tr(W_n\Sigma(P_n))$.
The following are equivalent:
  \begin{enumerate}
  \item[(i)] $\widehat V_{W,n}\ge0$ for every realization $Y_n\in\mathbb R^n$
        (realized non-negativity);
  \item[(ii)] $E_{P_n} \left[ \widehat V_{W,n} \right] \geq \theta_n$ on
        $\mathcal P_n(\mathbb R^n,\mathcal S_n)$ for every closed convex cone
        $\mathcal S_n\subseteq\mathrm{Sym}_n^+$;
  \item[(iii)] $W_n\succeq0$.
  \end{enumerate}
\end{proposition}

Beyond this similarity, the purpose and objective of the conservativeness are different. 
Corrections of the variance estimator are motivated by realized non-negativity: a quadratic form $Y_n'W_nY_n$ with $W_n\not\succeq0$ can return a negative number in a given sample, and the standard remedy enforces $\widehat V\ge0$ directly.
A realized-non-negativity fix acts on the \emph{estimate}: it adjusts the random variable $\widehat V_{W,n}$ so that no draw falls below zero, but it does not change which estimand $V_{W,n}(P_n)=\tr(W_n\Sigma)+\mu'W_n\mu$ is being targeted. 
When the plug-in is anticonservative, the problem is that $V_{W,n}(P_n)<\theta_n(P_n)$, i.e.\ $\mu'W_n\mu<0$, even though $V_{W,n}(P_n)$ itself remains positive because the covariance term $\tr(W_n\Sigma)=\theta_n$ dominates; a floor at zero leaves the gap $\mu'W_n\mu$ untouched. 
The correction developed in this paper acts instead on the \emph{estimand}: 
replacing $W_n$ with an $A_n$ that satisfies the conditions of \Cref{thm:conservative} changes the targeted quantity.

\section{Optimal Conservative Variance Estimation}
\label{sec:optimal-conservative}

This section develops criteria that define an optimal variance estimator. 
Section~\ref{sec:opt-problems} introduces these optimization criteria while  Section~\ref{sec:eigtruncation} shows that an eigenvalue truncation (ET) solution $A_n^\star$ solves all problems under some conditions. 
When the conditions for ET optimality fail, these criteria can nonetheless be formulated as conic programs that can be solved numerically --- the topic of \Cref{sec:pi-programs}.

\subsection{Three optimization problems}
\label{sec:opt-problems}
By \Cref{thm:conservative}, every uniformly conservative quadratic estimator on $\mathcal{P}_n(\mathcal{M}_n, \mathcal{S}_n)$ corresponds to a matrix $A_n \in \mathrm{Sym}_n$ satisfying (C1)--(C2). 
Among such estimators we seek one that is ``best". 
Three natural criteria arise: the first is a minimal correction criterion that ensures $A_n$ departs from $W_n$ as little as possible; the second is a level criterion that minimizes the variance estimand subject to validity, in the interest of the most powerful tests; and the third is the MSE, the common benchmark for the quality of an estimator.

The target $\theta_n = \tr(W_n\Sigma)$ is a known function of $W_n$, and the plug-in $A_n = W_n$ estimates it with mean bias $\mu'W_n\mu$.
The plug-in is anticonservative precisely when ${W}_n \not\succeq 0$ (\Cref{thm:conservative}); when it is conservative, no correction is warranted. 
Where a correction is needed, the natural object is the one that repairs anticonservativeness while disturbing the plug-in as little as possible. 
We measure the disturbance by the worst-case upward adjustment of $A_n$ relative to the plug-in $W_n$ over admissible mean directions,
\begin{equation}
\mathcal{R}_n(A_n)
:= \sup_{\mu \in \mathcal{M}_n,\ \|\mu\| = 1}\ \mu'(A_n - W_n)\mu.
\label{eq:risk-bias}
\end{equation}

\textbf{Criterion 1: Minimal Correction.} Let $P_{\mathcal{M}_n}$ denote the orthogonal projector onto $\mathcal{M}_n$. The minimal-correction problem is
\begin{equation}
\min_{A_n \in \mathrm{Sym}_n}\ \mathcal{R}_n(A_n)
\qquad \text{subject to}\quad
A_n - W_n \in \mathcal{S}_n^*,\ \ \
P_{\mathcal{M}_n} A_n P_{\mathcal{M}_n} \succeq 0,
\label{eq:opt-min-corr}
\end{equation}
and it depends on $(\mathcal{M}_n, W_n)$ alone, so its solution applies uniformly across the model class. 
In the optimization problem \eqref{eq:opt-min-corr} (C2) is imposed in the equivalent single-matrix form $P_{\mathcal{M}_n}A_nP_{\mathcal{M}_n}\succeq0$. 
The equivalence is immediate: for every $\mu\in\mathcal{M}_n$ one has $P_{\mathcal{M}_n}\mu=\mu$, so $\mu'A_n\mu=\mu'P_{\mathcal{M}_n}A_nP_{\mathcal{M}_n}\mu$, and ranging over $\mu\in\mathcal{M}_n$ shows (C2) holds if and only if the restriction of $A_n$ to $\mathcal{M}_n$ is positive semidefinite. 
The projected form is the usable one for computation: it is a single linear matrix inequality rather than an infinite family of scalar constraints.

\textbf{Criterion 2: Level.} The level of the estimand targeted by $A_n$ is its expectation
\[
V_{A,n}(\mu,\Sigma)=\E_{P_n}[Y_n'A_nY_n]
=\tr(A_n\Sigma)+\mu'A_n\mu
=\tr\!\big(A_n(\Sigma+\mu\mu')\big),
\]
and among conservative corrections one would like the one whose estimand is smallest, so as not to over-inflate the variance and lose power. 
The pointwise level problem at a given admissible $(\mu,\Sigma)$ is
\begin{equation}
\min_{A_n\in\mathrm{Sym}_n}\ \tr\!\big(A_n(\Sigma+\mu\mu')\big)
\qquad\text{s.t.}\qquad
A_n-W_n\in\mathcal S_n^*,\quad
P_{\mathcal M_n}A_nP_{\mathcal M_n}\succeq0.
\label{eq:opt-level-pointwise}
\end{equation}

\textbf{Criterion 3: MSE.} With $\mathrm{MSE}_n(A_n;\mu,\Sigma):=\E_{P_n}\!\big[(Y_n'A_nY_n-\theta_n)^2\big]$, the pointwise MSE problem at a given admissible $(\mu,\Sigma)$ is
\begin{equation}
\min_{A_n\in\mathrm{Sym}_n}\ \mathrm{MSE}_n(A_n;\mu,\Sigma)
\qquad\text{s.t.}\qquad
A_n-W_n\in\mathcal S_n^*,\quad
P_{\mathcal M_n}A_nP_{\mathcal M_n}\succeq0.
\label{eq:opt-mse-pointwise}
\end{equation}

Both \eqref{eq:opt-level-pointwise} and \eqref{eq:opt-mse-pointwise} reference the unknown $(\mu,\Sigma)$, and---unlike minimal correction---neither admits a worst-case formulation over the model: along $\Sigma=c\Sigma_0$ the MSE variance term grows like $c^2$ and the level term $\tr(A_n\Sigma)$ grows like $c$, so the supremum over the scale-free cone $\mathcal S_n$ diverges for every feasible $A_n$ and cannot rank candidates.
\Cref{sec:pi-programs} resolves this by replacing the worst case with an average case.

\begin{remark}[The bias and level criteria coincide]
\label{rem:bias-equals-level}
A natural objective to consider is to minimize the bias $|V_{A,n}(\mu, \Sigma) - \theta_n|$, but this problem coincides with the level problem. 
The bias of $A_n$ decomposes as $V_{A,n}(\mu,\Sigma)-\theta_n =\tr\!\big(A_n(\Sigma+\mu\mu')\big)-\tr(W_n\Sigma)$, and $\tr(W_n\Sigma)$ does not depend on $A_n$. 
Hence, minimizing the bias and minimizing the level are therefore the \emph{same} optimization problem.
This result is also intuitive: validity implies $V_{A,n}(\mu,\Sigma) \geq \theta_n$, so for a given $(\mu,\Sigma)$, the bias is simply $V_{A,n}(\mu, \Sigma) - \theta_n$, and minimizing that object is the same as minimizing  $V_{A,n}(\mu,\Sigma)$, subject to validity.
\end{remark}

\begin{remark}[A degenerate minimizer]
\label{rem:degenerate-minimizer}
Since \eqref{eq:opt-min-corr} reads $A_n$ only through its $\mathcal M_n$-block, it
cannot exclude estimators that are cheap on $\mathcal M_n$ but wild off it. Take
$\mathcal S_n^{\mathrm{diag}}$, $\mathcal M_n=\mathbf 1_n^\perp$, and $W_n=I_n$. Because $W_n=I_n$ acts
as the identity on $\mathcal M_n$, every feasible $A_n$ satisfies
\[
\mathcal R_n(A_n)
=\sup_{\mu\in\mathcal M_n,\ \|\mu\|=1}\mu'(A_n-I_n)\mu
=\Big(\sup_{\mu\in\mathcal M_n,\ \|\mu\|=1}\mu'A_n\mu\Big)-1\;\ge\;-1,
\]
the inequality by (C2), with equality if and only if $\mu'A_n\mu=0$ for every
$\mu\in\mathcal M_n$. The value $-1$ is therefore the minimum of
\eqref{eq:opt-min-corr}, and it is attained by $A_n=\mathbf 1_n\mathbf 1_n'$, which
is feasible --- (C1) holds since $\diag(\mathbf 1_n\mathbf 1_n'-I_n)=0$, and (C2)
since $P_{\mathcal M_n}\mathbf 1_n=0$ --- and which yields the trivial estimator
$\widehat V_{A,n}=(\sum_i Y_{n,i})^2$. 
The minimal-correction criterion thus selects
an estimator with no sampling content at all, strictly preferring it to the
truncation; it is insufficient for selecting an estimator, and the variance
discipline of the MSE problem~\eqref{eq:opt-mse-pointwise} is what rules such
minimizers out. 
\end{remark}

The MSE criterion controls sampling variance directly, eliminating the degeneracy illustrated by $(\sum_i Y_{n,i})^2$. The cost is that the objective involves the third- and fourth-moment tensors.
Specifying these requires distributional assumptions beyond the second-moment structure encoded by $(\mathcal{M}_n, \mathcal{S}_n)$.

\subsection{The eigenvalue-truncation solution}
\label{sec:eigtruncation}

The weight matrix $W_n$ has the following spectral decomposition
\[
W_n=U_n\Lambda_nU_n',
\qquad
\Lambda_n=\diag(\lambda_{n,1},\dots,\lambda_{n,n}).
\]
Write $\lambda^{+}:=\max\{\lambda,0\}$
and $\lambda^{-}:=\max\{-\lambda,0\}$, set
$\Lambda_n^{\pm}:=\diag(\lambda_{n,1}^{\pm},\dots,\lambda_{n,n}^{\pm})$, and define
\[
W_n^{+}:=U_n\Lambda_n^{+}U_n',
\qquad
W_n^{-}:=U_n\Lambda_n^{-}U_n',
\qquad
\lambda_n^{\min}:=\min_j\lambda_{n,j}=\lambda_{\min}(W_n),
\]
so $W_n=W_n^{+}-W_n^{-}$, both terms are positive semidefinite, and $\lambda_n^{\min}$
is the smallest eigenvalue of $W_n$. The \emph{eigenvalue-truncation} (ET) weight matrix
is:
\begin{equation}
A_n^\star :=W_n^{+}
=U_n\Lambda_n^{+}U_n' .
\label{eq:A-star}
\end{equation}
It is identical to $W_n$ except that it truncates the negative eigenvalues at zero.
$A_n^\star$ depends only on the known $W_n$; it never references $\mathcal S_n$ and never
requires an estimate of $\Sigma$.
Due to \Cref{thm:conservative}, $A_n^\star$ is valid for any $(\mathcal M_n, \mathcal S_n)$, since $W^+_n$ is positive semidefinite. 

The following definitions and assumptions are used to characterize optimality.

\begin{definition}[Spectral class and aligned cone]
\label{def:spectral-aligned}
Fix an orthonormal eigenbasis $\{q_{n,j}\}_{j=1}^n$ of $W_n$, so that
$W_n=\sum_{j=1}^n\lambda_{n,j}\,q_{n,j}q_{n,j}'$ (where $q_{n,j} = U_{n,\cdot j}$), and write
\[
\Delta(a):=\sum_{j=1}^{n}a_j\,q_{n,j}q_{n,j}',
\qquad a\in\R^n .
\]
The \emph{spectral class} is
\[
\mathcal A_n^{\mathrm{sp}}
:=\bigl\{A_n(a):=W_n+\Delta(a)\ :\ a\in\R^n\bigr\},
\]
and the \emph{aligned cone} is
\[
\mathcal S_n^{W\text{-}\mathrm{diag}}
:=\bigl\{\Delta(s):s\in\R^n_+\bigr\},
\]
the covariances commuting with $W_n$. 
\end{definition}

$\mathcal S_n^{W\text{-}\mathrm{diag}}$ is generated by the eigenprojectors
$q_{n,j}q_{n,j}'$, so
\begin{equation}
\mathcal S_n^{W\text{-}\mathrm{diag}}\subseteq\mathcal S_n
\iff
q_{n,j}q_{n,j}'\in\mathcal S_n\ \text{ for every }j .
\label{eq:containment}
\end{equation}
The containment~\eqref{eq:containment} is the operative restriction on the cone in (ii) and (iii) below. It is a
demand on $\mathcal S_n$ from below, and is inherited by every larger cone.

The MSE criterion is a fourth-order object, so it cannot be evaluated from
$(\mathcal M_n,\mathcal S_n)$ alone. We impose the following shape restriction, which makes
the MSE a function of $(\mu,\Sigma)$ only.

\begin{assumption}[No skewness, Gaussian fourth-cumulant structure]
\label{asmp:no-skew}
For every $P_n\in\mathcal P_n(\mathcal M_n,\mathcal S_n)$ the vector $Y_n$ has finite fourth
moments, and $Z:=Y_n-\mu$ satisfies
\[
\E_{P_n}[Z_iZ_jZ_k]=0\quad\text{for all }i,j,k,
\qquad
\kappa_{ijkl}(P_n)=0\quad\text{for all }i,j,k,l,
\]
where $\kappa_{ijkl}(P_n)=\E_{P_n}[Z_iZ_jZ_kZ_l]-\E[Z_iZ_j]\E[Z_kZ_l]
-\E[Z_iZ_k]\E[Z_jZ_l]-\E[Z_iZ_l]\E[Z_jZ_k]$.
\end{assumption}

The no-skewness condition is implied by symmetry of $Z$ about zero, $Z\overset{d}{=}-Z$; the
Gaussian fourth-cumulant condition is implied by Gaussianity, and for elliptical
distributions $\kappa^{(4)}$ is proportional to the Gaussian tensor scaled by a kurtosis
parameter. \Cref{asmp:no-skew} restricts the \emph{shape} of $Z$, not its location $\mu$ or
scale $\Sigma$, both of which remain arbitrary within
$\mathcal M_n\times\mathcal S_n$.

\begin{theorem}[Eigenvalue truncation: validity and optimality]
\label{thm:eig-truncation}
Let $A_n^\star$ be as in~\eqref{eq:A-star}, and assume $\mathcal M_n=\R^n$. 
\begin{enumerate}
\item[(i)] \textbf{Minimal correction.} Suppose $\mathcal{S}_n \ne \{ 0 \}$. Then, $A_n^\star=W_n^+$   solves~\eqref{eq:opt-min-corr} with optimal value $\mathcal R_n(A_n^\star)
=(-\lambda_n^{\min})_+$.
\end{enumerate}
Parts (ii)(a)--(b) and (iii) assume in
addition the containment~\eqref{eq:containment}.
\begin{enumerate}
\item[(ii)] \textbf{Level.}
  \begin{enumerate}
  \item[(a)] For every $\Psi=\Sigma+\mu\mu'\succeq0$, $A_n^\star$ solves the level
  problem~\eqref{eq:opt-level-pointwise} over the uniformly conservative members of
  $\mathcal A_n^{\mathrm{sp}}$. 
  This solution is unique if and only if $\psi_j:=q_{n,j}'\Psi q_{n,j}>0$ for every $j$.
  \item[(b)] If in addition $\Psi\in\mathcal S_n^{W\text{-}\mathrm{diag}}$, then $A_n^\star$ minimizes $\tr(A_n\Psi)$ over \emph{all} uniformly conservative $A_n\in\Sym_n$,
spectral or not.
  \item[(c)] Let $\mathcal S_n$ be closed.
  If $q_{n,k}q_{n,k}'\notin\mathcal S_n$ for some $k$ with $\lambda_{n,k}>0$, then some
  uniformly conservative $A_n\in\Sym_n$ satisfies $\tr(A_n\Psi)<\tr(A_n^\star\Psi)$ at
  $\Psi=q_{n,k}q_{n,k}'\succeq0$.
  \end{enumerate}
\item[(iii)] \textbf{MSE.} Suppose \Cref{asmp:no-skew} holds. Then for every $\mu\in\R^n$ and
every $\Sigma\in\mathcal S_n^{W\text{-}\mathrm{diag}}$,
\[
\mathrm{MSE}_n(A_n^\star;\mu,\Sigma)\le\mathrm{MSE}_n(A_n;\mu,\Sigma)
\qquad\text{for every uniformly conservative }A_n\in\mathcal A_n^{\mathrm{sp}},
\]
so $A_n^\star$ solves~\eqref{eq:opt-mse-pointwise} pointwise; the minimizer is
unique if $s_j:=q_{n,j}'\Sigma q_{n,j}>0$ for every $j$.
\end{enumerate}
\end{theorem}

In (ii) and (iii), the spectral class $\mathcal{A}^{\mathrm{sp}}_n$ is the set of weight matrices built from the known $W_n$ by adjusting its eigenvalues but not its eigenvectors: every estimator in \Cref{sec:examples-and-optimal} lies in it, so it is not restrictive.
Within it, ET is the Loewner-least element. 
By (ii)(b), at any $\Psi$ diagonal in the eigenbasis of $W_n$, no uniformly conservative $A_n$ undercuts ET. 
What a non-spectral $A_n$ can do is beat ET at \emph{particular} $(\mu,\Sigma)$; capturing that gain means orienting weight toward the unknown second moment and, for MSE, toward its split into $\mu\mu'$ and $\Sigma$, which \Cref{thm:impossibility} forbids learning.
Parts (ii)(a) and (iii) differ only in that (iii) claims pointwise optimality in $\mathcal{S}_n^{\mathrm{W-diag}}$ while (ii) claims pointwise optimality for every $\Psi \succeq 0$.

\begin{remark}[The conditions constrain the cone]
\label{rem:richness-detectability}
The constraint \eqref{eq:containment} asks the cone to place mass on each eigendirection of $W_n$, which can be restrictive. 
Let the covariance model be a zero pattern, $\mathcal S_n=\{\Sigma\succeq0:\Sigma_{ij}=0$ whenever units $i$ and $j$ are not
permitted to covary$\}$. By~\eqref{eq:containment}, what matters is whether each $q_{n,j}q_{n,j}'$ vanishes at every forbidden cell.
Under two-way clustering, units are indexed by $(g,h)$ with $g\le G$ and $h\le H$, $G,H\ge2$;
two units may covary exactly when they share $g$ or share $h$.
The top eigenvector is $q=\mathbf 1_n/\sqrt n$, and its projector $\mathbf 1_n\mathbf 1_n'/n$ is nonzero at
\emph{every} cell, including the pairs sharing neither index. Hence $qq'\notin\mathcal S_n$, \eqref{eq:containment} fails, and by
\Cref{thm:eig-truncation}(ii)(c) some uniformly conservative $A_n$ strictly undercuts ET at
$\Psi=qq'$. Under a bare two-way cone, then, ET remains uniformly conservative and
minimal-correction optimal, but is not level- or MSE-optimal.
\end{remark}

A researcher unwilling to buy optimality with those covariance restrictions can instead enlarge the cone until the containment holds. 
Call a convex cone $\mathcal T\subseteq\Sym_n^+$ an \emph{admissible enlargement} of $\mathcal S_n$ if $\mathcal S_n\subseteq\mathcal T$ and $\mathcal T$ satisfies~\eqref{eq:containment}. 
Adding the aligned cone produces the smallest enlargement, formalized in the next proposition.

\begin{proposition}[The minimal aligned enlargement]
\label{prop:union-cone}
Let $\mathcal M_n=\R^n$, let $\mathcal S_n\subseteq\Sym_n^+$ be a convex cone, and let
\[
\mathcal S_n^{\oplus}:=\mathcal S_n+\mathcal S_n^{W\text{-}\mathrm{diag}}
=\{\Sigma_1+\Sigma_2:\Sigma_1\in\mathcal S_n,\ 
\Sigma_2\in\mathcal S_n^{W\text{-}\mathrm{diag}}\}
\]
denote their Minkowski sum.
\begin{enumerate}
\item[(i)] \textbf{Minimality.} $\mathcal S_n^{\oplus}$ is an admissible enlargement of
$\mathcal S_n$, and $\mathcal S_n^{\oplus}\subseteq\mathcal T$ for every admissible enlargement
$\mathcal T$.
\item[(ii)] \textbf{The dual price.} $(\mathcal S_n^{\oplus})^*
=\mathcal S_n^*\cap(\mathcal S_n^{W\text{-}\mathrm{diag}})^*$, where
\[
(\mathcal S_n^{W\text{-}\mathrm{diag}})^*
=\bigl\{D\in\Sym_n:\ q_{n,j}'Dq_{n,j}\ge0\ \text{ for every }j\bigr\} ,
\]
so $A_n$ is uniformly conservative on $\mathcal P_n(\R^n,\mathcal S_n^{\oplus})$ if and only if
it is uniformly conservative on $\mathcal P_n(\R^n,\mathcal S_n)$ \emph{and}
$q_{n,j}'(A_n-W_n)q_{n,j}\ge0$ for every $j$. Moreover
$\mathcal T^*\subseteq(\mathcal S_n^{\oplus})^*$ for every admissible enlargement
$\mathcal T$.
\end{enumerate}
\end{proposition}

Optimality is then a corollary. Since $\mathcal S_n^{\oplus}$ satisfies \eqref{eq:containment}, \Cref{thm:eig-truncation}(ii)(a)--(b) and (iii) apply to it verbatim, so $A_n^\star=W_n^{+}$ solves the level problem at every $\Psi\succeq0$ within $\mathcal A_n^{\mathrm{sp}}$, and, under \Cref{asmp:no-skew}, the MSE problem at every $\mu\in\R^n$ and every $\Sigma\in\mathcal S_n^{W\text{-}\mathrm{diag}}$. 
The enlargement changes the constraint set, not the estimator: by \eqref{eq:A-star}, $A_n^\star$ never references $\mathcal S_n$, so any convergence result for $Y_n'A_n^\star Y_n$ established under the dependence environment that generates $\mathcal S_n$ continues to hold unchanged.

\begin{remark}[Optimality over a restricted dual cone]
\label{rem:enlargement-cost}
Validity is unaffected. What the enlargement changes is the
feasible set of \eqref{eq:opt-min-corr}, \eqref{eq:opt-level-pointwise} and
\eqref{eq:opt-mse-pointwise}: by \Cref{prop:union-cone}(ii) the constraint (C1) tightens from
$A_n-W_n\in\mathcal S_n^*$ to
\[
A_n-W_n\in\mathcal S_n^*
\quad\text{and}\quad
q_{n,j}'(A_n-W_n)q_{n,j}\ge0\ \text{ for every }j ,
\]
which discards exactly the competitors that undercut ET along eigendirections \Cref{thm:eig-truncation}(ii)(c) constructs. 
Optimality of ET is therefore optimality \emph{over a restricted dual cone}, and by \Cref{prop:union-cone} the restriction is the smallest one that delivers it: any other admissible enlargement discards more competitors.


An $A_n$ beating ET under $\mathcal S_n$ places weight on the cells $\mathcal S_n$ \emph{declares} to be exactly zero. 
The extra dual constraint of \Cref{prop:union-cone}(ii), $q_{n,j}'(A_n-W_n)q_{n,j}\ge0$ for every $j$, is precisely the inequality it violates. 
Such an estimator draws its advantage from being valid if the zero pattern holds exactly, but fragile to misspecification of covariances. Declining that advantage is what $\mathcal S_n^{\oplus}$ formalizes.
\end{remark}

The ET variance estimator is shown to converge to the intended variance estimand under some primitive conditions in \Cref{lem:quadratic_LLN} in the appendix. 
Hence, when coupled with a central limit theorem, the variance estimator can be used for robust and optimal inference. 
These limit theorems are often derived under particular dependence structures, so the formal statement is deferred to \Cref{sec:application_panel_dep}.
We fix $A_n^\star=W_n^+$ as the closed-form benchmark. 

\subsection{Average-case criteria }
\label{sec:pi-programs}

Outside the regime where optimality is established, no correction dominates pointwise: the MSE- and level-minimizing $A_n$ move with $(\mu,\Sigma)$, orienting weight toward the mass of $\Sigma$ in a way Theorem~\ref{thm:impossibility} forbids learning. 
Similarly, the minimal correction problem may not be minimized at ET when $\mathcal M_n \ne \mathbb R^n$. 
A worst-case selection is unavailable as the entire identification problem is that the second moment $\E[Y_nY_n']=\Sigma+\mu\mu'$ does not reveal how its mass splits between $\mu\mu'$ and $\Sigma$. 
There is no canonical adversary, and selecting a unique estimator requires committing to \emph{how} the residual uncertainty in $(\mu,\Sigma)$ is weighted.

We make that weighting explicit. 
Let $\pi$ be a researcher-declared probability measure on $\mathcal M_n\times\mathcal S_n$ with finite fourth moments, and define the $\pi$-averaged MSE and the $\pi$-averaged level,
\begin{align}
\overline{\mathrm{MSE}}_n(A_n;\pi)
&:=\int \mathrm{MSE}_n(A_n;\mu,\Sigma)\,d\pi(\mu,\Sigma),
\\
\overline V_n(A_n;\pi)
&:=\int \tr\!\big(A_n(\Sigma+\mu\mu')\big)\,d\pi(\mu,\Sigma)
=\tr(A_n\Psi_\pi),
\end{align}
where $\Psi_\pi:=\E_\pi[\Sigma+\mu\mu']\succeq0$ is the prior second moment of $Y_n$. 
The associated programs are
\begin{align}
\min_{A_n\in\mathrm{Sym}_n}\ \overline{\mathrm{MSE}}_n(A_n;\pi)
&\qquad\text{s.t.}\qquad
A_n-W_n\in\mathcal S_n^*,\quad
P_{\mathcal M_n}A_nP_{\mathcal M_n}\succeq0,
\label{eq:opt-mse-bayes}\\[2pt]
\min_{A_n\in\mathrm{Sym}_n}\ \tr(A_n\Psi_\pi)
&\qquad\text{s.t.}\qquad
A_n-W_n\in\mathcal S_n^*,\quad
P_{\mathcal M_n}A_nP_{\mathcal M_n}\succeq0.
\label{eq:opt-level-bayes}
\end{align}
The measure $\pi$ requires no estimation of $\mu$ or $\Sigma$ and never invokes the extreme rays of $\mathcal S_n$. 
Averaging names no adversary, so both \eqref{eq:opt-mse-bayes} and \eqref{eq:opt-level-bayes} remain finite.
Notably, the choice of $\pi$ is irrelevant for validity, because the constraints are satisfied by the $A_n$ that solves the programs even if the $\pi$ is misspecified. 

The minimal-correction criterion does not require a declared $\pi$, because its objective $\mathcal R_n(A_n)=\sup_{\mu\in\mathcal M_n,\|\mu\|=1}\mu'(A_n-W_n)\mu$ depends only on $(\mathcal M_n,W_n)$, both known, and its worst case is taken over the compact unit sphere of $\mathcal M_n$.
The optimization problem for this criterion directly is also conic. 

\begin{proposition}[Conicity of the optimization problems]
\label{prop:conicity}
Suppose $\mathcal{S}_n$ is a closed convex cone and $\mathcal{M}_n$ is a linear
subspace.
\begin{enumerate}
\item[(i)] The minimal-correction problem~\eqref{eq:opt-min-corr} is a convex conic
program in $A_n$.
\item[(ii)] Let $\pi$ be a probability measure on $\mathcal{M}_n\times\mathcal{S}_n$
with finite fourth moments, and suppose that for $\pi$-almost every
$(\mu,\Sigma)$ the admissible moment tensors $(m^{(3)},\kappa^{(4)})$ are such
that $\mathrm{MSE}_n(A_n;\mu,\Sigma)$ is convex in $A_n$. Then the $\pi$-averaged
MSE program~\eqref{eq:opt-mse-bayes} is a convex conic program in $A_n$.
\item[(iii)] Let $\pi$ be a probability measure on
$\mathcal{M}_n\times\mathcal{S}_n$ with finite second moments. Then the
$\pi$-averaged level program~\eqref{eq:opt-level-bayes} is a \emph{linear} conic
program in $A_n$.
\end{enumerate}
\end{proposition}

Conicity is practically attractive: any specific instance of these three programs can be solved with interior-point solvers once the problem data $(W_n, \mathcal{M}_n, \mathcal{S}_n)$ and the moment-tensor class are specified. 
The conic formulation is the fallback when no closed form is available, say, when the conditions of \Cref{thm:eig-truncation} do not hold, and restores a computable optimum.

Under the conditions of \Cref{thm:eig-truncation}, ET is also the optimal solution for these optimization problems regardless of the choice of $\pi$. 
Since ET is pointwise optimal for \emph{every} $\pi$, it continues to be optimal after integrating over $\pi$. 
The average-case criterion is thus a single criterion across regimes.

However, unlike the ET solution, it is in general more difficult to show that the variance estimators obtained from solving numerical conic programs converge to the intended estimand.

\section{Examples and Optimal Corrections}
\label{sec:examples-and-optimal}

This section explains how the framework's optimization problem applies to standard and new variance estimators. 
Each subsection states the weight matrix $W_n$, the admissible cone $\mathcal{S}_n$, and the optimal correction $A_n^\star$.

\subsection{One-way cluster-robust estimators}
\label{sec:oneway-cluster}

With observations partitioned into clusters $G_1, \ldots, G_{C_n}$ and indicator vectors $B_c \in \mathbb{R}^n$, the one-way cluster-robust estimator is
\[
\widehat{V}_{\clu, n} = \sum_{c=1}^{C_n}\!\left(\sum_{i \in G_c} Y_{n,i}\right)^{\!2} = Y_n' W_n^{\clu} Y_n, \qquad W_n^{\clu} = \sum_{c=1}^{C_n} B_c B_c'.
\]
Each $B_c B_c'$ is positive semidefinite, so $W_n^{\clu}$ is globally PSD, and the plug-in is uniformly conservative for every admissible mean space.
Hence $A_n^\star = W_n^{\clu}$, and heterogeneous means cannot induce anticonservativeness when the weight matrix is globally positive semidefinite. 
This result shows how the plug-in one-way cluster robust variance estimator is not only conservative in design-based settings (as established in \citet{abadie2023should}), it is also optimal.

\subsection{Time series: HAC estimators}
\label{sec:hac-revisit}

A generic HAC estimator is the quadratic form
\[
\widehat{V}_{\HAC,n}
= \sum_{i=1}^n\sum_{j=1}^n K\!\left(\frac{i-j}{b_n}\right)Y_{n,i}Y_{n,j}
= Y_n'W_n^{\HAC}Y_n,
\qquad
(W_n^{\HAC})_{ij}=K\!\left(\frac{i-j}{b_n}\right),
\]
with $W_n^{\HAC}$ symmetric Toeplitz. 
Whether the plug-in $A_T^\star=W_n^{\HAC}$ is already optimal turns entirely on the \emph{spectral window} of the kernel. 
For a symmetric banded Toeplitz matrix $B$ with value $c_k$ on its $k$-th diagonal, the spectral window is the generating function
\[
f(\omega):=c_0+2\sum_{k\ge1}c_k\cos(k\omega),\qquad\omega\in[-\pi,\pi],
\]
the Fourier transform of the kernel weights. 
Two facts connect $f$ to the eigenvalues of $B$: every eigenvalue lies in the range of $f$ (each Rayleigh quotient is a weighted average of $f$ values), and by Grenander--Szeg\H{o} the empirical eigenvalue distribution converges to that of $f$. 
The sign of $f$ therefore governs whether the plug-in is conservative. 
We contrast two kernels.

\subsubsection{Naive m-dependence plug-in}
\label{sec:m-dep-naive}

Consider a scalar time series $\{Y_t\}_{t=1}^T$ that is $m$-dependent, so $\mathrm{Cov}(Y_t,Y_{t+k})=0$ for $|k|>m$. 
The variance of the sum is
\[
\mathrm{Var}\!\left(\sum_{t=1}^T Y_t\right)
= \sum_t \mathrm{Var}(Y_t)
+ \sum_{k=1}^m \sum_t \big[\mathrm{Cov}(Y_t,Y_{t+k})+\mathrm{Cov}(Y_{t+k},Y_t)\big].
\]
The naive plug-in replaces each covariance with the product $Y_tY_{t+k}$ and weights all admissible lags equally,
\[
\widehat{V}^{\mathrm{naive}}_T
= \sum_t Y_t^2 + \sum_{k=1}^m \sum_t \big(Y_t Y_{t+k}+Y_{t+k}Y_t\big)
= Y'B_m Y,
\]
where $B_m\in\mathrm{Sym}_T$ is the symmetric banded Toeplitz matrix with ones on bands $0,1,\dots,m$. 
Its spectral window is the Dirichlet kernel of order $m$,
\[
f_{\mathrm{Dir}}(\omega)
= 1+2\sum_{k=1}^m\cos(k\omega)
= \frac{\sin\!\big((m+\tfrac12)\omega\big)}{\sin(\omega/2)},
\]
which oscillates in sign, vanishing at $\omega=2\pi j/(2m+1)$, $j=1,\dots,m$, and taking negative values between successive zeros. 
A positive fraction of the eigenvalues of $B_m$ are negative, so $B_m$ is indefinite. 
Two consequences follow. 
First, the realized $Y'B_m Y$ can be negative in a given sample regardless of the mean---the classical defect that motivates a positive-semidefinite kernel. 
Second, under the agnostic mean space $\mathcal{M}_n = \mathbb{R}^T$ the negative eigendirections lie inside $\mathcal{M}_n$, so the plug-in estimand is anticonservative for means loading on them. 
Due to \Cref{thm:conservative}, the variance estimand can be anticonservative --- an example is provided in \Cref{sec:panel_dep_appendix}.
The eigenvalue truncation $A_T^\star = B_m^+$ addresses both issues.

The optimal correction is the positive-part operator
\[
A_T^\star=B_m^+=U\max(\Lambda,0)U',
\]
where $B_m=U\Lambda U'$ is the spectral decomposition. In the frequency domain, the optimal estimator has spectral window $f^\star(\omega)=\max\{f_{\mathrm{Dir}}(\omega),0\}$, modifying $B_m$ only on the frequency band where $f_{\mathrm{Dir}}<0$.

\subsubsection{Newey--West with Bartlett kernel}
\label{sec:newey-west-revisit}

The Newey--West estimator replaces uniform weights by linearly declining Bartlett weights $\omega(k,M)=1-k/(M+1)$ for $|k|\le M$,
\[
\widehat{V}^{\mathrm{NW}}_T=Y'C_M Y,\qquad (C_M)_{ts}=\omega(|t-s|,M),
\]
with $C_M\in\mathrm{Sym}_T$ banded Toeplitz. 
Its spectral window is the Fej\'er kernel,
\[
f_{\mathrm{Fej}}(\omega)
= \frac{1}{M+1}\!\left(\frac{\sin((M+1)\omega/2)}{\sin(\omega/2)}\right)^{\!2}\ge0,
\]
the squared modulus of a Dirichlet kernel, hence nonnegative everywhere.
Because every eigenvalue of a symmetric Toeplitz matrix lies in the range of its spectral window, nonnegativity of $f_{\mathrm{Fej}}$ forces $C_M\succeq0$ exactly, in every finite sample---the positive-semidefiniteness for which the Bartlett kernel was introduced \citep{newey1987asimple}. 
No truncation is ever required: $A_T^\star=C_M$, and the identification problem of this paper does not bind in the Bartlett--HAC setting.

\subsection{Two-way clustering}
\label{sec:twoway_optimal}

We work in a balanced panel with $G$ groups and $T$ periods, one observation per cell, so $n = GT$, and identify $\mathbb{R}^n \cong \mathbb{R}^T \otimes \mathbb{R}^G$.
With row clusters indexed by $g$ and column clusters by $t$, the Cameron--Gelbach--Miller (CGM) weight matrix is
\[
W_n^{2\text{-way}} = W_n^G + W_n^T - W_n^{G\cap T},
\qquad
W_n^G = J_T \otimes I_G,\quad
W_n^T = I_T \otimes J_G,\quad
W_n^{G\cap T} = I_n,
\]
where $J_G = \mathbf{1}_G\mathbf{1}_G'$ and $J_T = \mathbf{1}_T\mathbf{1}_T'$.
Each component is positive semidefinite, but the subtraction renders $W_n^{2\text{-way}}$ indefinite, so under heterogeneous means the plug-in is anticonservative \citep{xu2024clustering}---a direct instance of the converse in Theorem~\ref{thm:conservative} at $\mathcal{M}_n=\mathbb{R}^n$.
Let
\[
P_{00} := I_n - \frac{W_n^G}{T} - \frac{W_n^T}{G} + \frac{J_n}{n},
\qquad J_n = \mathbf{1}_n\mathbf{1}_n',
\]
be the orthogonal projector implementing two-way demeaning,
$(P_{00}\,y)_{gt} = y_{gt} - \bar y_{\cdot t} - \bar y_{g\cdot} + \bar y_{\cdot\cdot}$, and the associated subspace $\mathcal{S}_{00} := \range (P_{00})$.

\begin{corollary}[Optimal two-way correction]
\label{cor:twoway_optimal}
Under $\mathcal{M}_n = \mathbb{R}^n$, the eigenvalue-truncation estimator
of Theorem~\ref{thm:eig-truncation} is
\[
A_n^\star = (W_n^{2\text{-way}})^+ = W_n^{2\text{-way}} + P_{00},
\]
equivalently $Y_n' A_n^\star Y_n = Y_n' W_n^{2\text{-way}} Y_n +
\sum_{g,t}\big(Y_{gt} - \bar Y_{\cdot t} - \bar Y_{g\cdot} + \bar Y_{\cdot\cdot}\big)^2$.
\end{corollary}

The CGM matrix subtracts $W_n^{G\cap T} = I_n$ to undo the double-counting of cells entering both the row-cluster and column-cluster sums. 
Consequently, on the two-way within subspace (vectors $v$ orthogonal to every row and column mean), we have $W_n^{2\text{-way}}v = -v$, resulting in a negative eigenvalue of $-1$. 
Truncating its eigenvalue to zero adds exactly the projector $P_{00}$ onto that subspace. 
Since the within subspace is precisely where heterogeneous means generate anticonservativeness (Theorem~\ref{thm:conservative}), the correction places nonnegative mass only where the plug-in was deficient.

CGM2 uses weight matrix $W_n^G + W_n^T = A_n^\star + (I_n - P_{00})$. 
Since $I_n - P_{00}\succeq0$, CGM2 is uniformly conservative but non-optimal: it overshoots $A_n^\star$ on the between-group and between-time directions, where no correction is needed. 
This overshoot is the same order as the correction when means carry between-cluster heterogeneity, so there is a finite-sample effect, visible in the simulations of Section~\ref{sec:simulation}, where ET delivers tighter intervals than the conservative C2NW (\Cref{sec:two-estimators}) while retaining size control.

\subsection{Two-way clustering with serial correlation}
\label{sec:chs-optimal}

We now allow cross-cluster serial correlation on the time dimension. 
In the balanced panel ($n = GT$), the Chiang--Hansen--Sasaki (CHS) plug-in is
\[
W_n^{\mathrm{CHS}} = J_T \otimes I_G + C_M \otimes (J_G - I_G),
\]
where $C_M \in \mathbb{R}^{T\times T}$ is the Bartlett kernel matrix,
$(C_M)_{ts} = \omega(|t-s|, M)$ with $\omega(m,M) = \max\{1 - m/(M+1), 0\}$.
This result is stated in \Cref{lem:chs-matrix-form} of \Cref{sec:panel_dep_appendix}. 
Let $P_G := J_G/G$ and $Q_G := I_G - P_G$ be the group-mean and
group-demeaning projectors, and define the indefinite $T\times T$ block
\[
B := J_T - C_M,
\qquad
B^- := (C_M - J_T)_+ \succeq 0,
\]
the negative eigenmass of $B$. Under $\mathcal{M}_n = \mathbb{R}^n$,
$A_n^\star = (W_n^{\mathrm{CHS}})^+$ by Theorem~\ref{thm:eig-truncation}.
Due to \Cref{thm:conservative}, the naive CHS variance estimator can be anticonservative, as exemplified in \Cref{sec:panel_dep_appendix}.

\begin{corollary}[Optimal correction under serial correlation]
\label{cor:chs_optimal}
\[
A_n^\star = (W_n^{\mathrm{CHS}})^+ = W_n^{\mathrm{CHS}} + B^- \otimes Q_G,
\]
equivalently, with $\check Y_{g,t} := Y_{g,t} - \bar Y_{\cdot t}$ the
group-demeaned (cross-sectionally centered) score,
\[
Y_n' A_n^\star Y_n
= Y_n' W_n^{\mathrm{CHS}} Y_n
+ \sum_{g=1}^G \check Y_{g,\cdot}'\, B^-\, \check Y_{g,\cdot}.
\]
\end{corollary}

Split the cross-section into the group-mean direction $P_G$ and its complement $Q_G$; the CHS weight acts block-diagonally across them. 
On the group-symmetric block the weight is $J_T + (G-1)C_M$, a nonnegative combination of two PSD matrices, hence already conservative---no correction. 
On the group-demeaned block the weight is $J_T - C_M$: the rank-one temporal aggregator $J_T$, concentrated on the constant-in-time direction, is partly canceled by the Fej\'er kernel $C_M$, producing negative eigenvalues. 
Only this block needs repair, and the truncation adds back precisely its negative eigenmass $B^-$, supported on the group-demeaned scores. 
The correction is thus a single $T\times T$ eigendecomposition applied within clusters. 
At $M = 0$, $C_M = I_T$ and $B^- = I_T - J_T/T$ is the time-demeaning projector, so $B^-\otimes Q_G$ reduces to two-way demeaning and we recover \Cref{cor:twoway_optimal}.


\begin{remark}[EVC floors the estimate; ET repairs the estimand]
\label{rem:evc-vs-et}
\citet{chiang2024standard} render their plug-in non-negative through an
eigenvalue correction (EVC) that replaces the negative eigenvalues of the
\emph{realized} meat $\widehat\Omega_{NT}$ by zero. In the notation here this
acts on the data-space realization $Y_n'W_n^{\mathrm{CHS}}Y_n$ (scalar case:
$\max\{Y_n'W_n^{\mathrm{CHS}}Y_n,0\}$), addressing (i) of
\Cref{prop:estimand-vs-estimate}. It does not act on $W_n^{\mathrm{CHS}}$, so it cannot
address (ii). At the plug-in covariance the bias-free part
$\tr(W_n^{\mathrm{CHS}}\Sigma_n)>0$,
so the estimand
$V_{\mathrm{CHS}}=\tr(W_n^{\mathrm{CHS}}\Sigma_n)+\mu_n'W_n^{\mathrm{CHS}}\mu_n$ is positive and
$\widehat\Omega_{NT}\to V_{\mathrm{CHS}}>0$, leaving the correction
asymptotically inactive---yet $V_{\mathrm{CHS}}<V_{\mathrm{true}}$ whenever
$\mu_n'W_n^{\mathrm{CHS}}\mu_n<0$ (the counterexample of
\Cref{sec:panel_dep_appendix}). 
\end{remark}

\section{Inference under Panel Dependence}
\label{sec:application_panel_dep}

This section develops the inferential machinery needed to implement variance estimators in this setting of \Cref{sec:chs-optimal}.

\subsection{Dependence framework}
\label{sec:dependence-framework}

Consider a two-way cluster dependent setting, where one dimension exhibits cross-cluster weak dependence. 
To keep the exposition concrete, suppose we have a panel with both a cross sectional and a time series component. 
The time series component contains time indices $t \in \{1, 2, \cdots, T\}$ for every observation, and observations within the same time period may be arbitrarily correlated. 
The cross sectional component can be partitioned on some dimension $G$, with cluster indices $g \in \{1, 2, \cdots, G\}$ such that every observation within the same $g$ cluster may be arbitrarily correlated. 
Let $i \in \{1, 2, \cdots, n\} =: N_n$ index observations, $g(i)$ denote its cross sectional cluster, and $t(i)$ denote its time period.
For example, in a canonical balanced panel with each cross sectional unit observed in every time period, we have $i = (g, t)$. 
While the setting in this paper permits unbalanced panels, the balanced panel with $G \asymp T$ is used as a running example.

There is a triangular array of multivariate random vectors $\{Y_{n,i}\}, n \geq 1, Y_{n,i} \in \mathbb{R}^v$, where $n$ denotes the number of observations. 
Define the moment of this random vector as $\|Y_{n,i}\|_p = (E[\|Y_{n,i}\|^p])^{1/p}$. 
All observations $i$ may be serially correlated. 
The distance $d_n(i, j)$ between any two observations $i$ and $j$ is $0$ if they share a cluster on any dimension, and the difference in their time index otherwise.

Let $N_t^T$ denote the set of observations in time $t$, $N_g^G$ denote the set of observations in cluster $g$, and $N_{t,g}^{T \cap G} := N_t^T \cap N_g^G$ the set of observations in both time $t$ and cluster $g$. 
Define the following objects to measure how concentrated observations are across clusters and time:
\begin{align*}
\delta_n^\partial(s; k) &:= \frac{1}{n} \sum_{i \in N_n} \left(|N_{t(i)+s}^T| + |N_{t(i)-s}^T| + \mathbf{1}\{s = 0\}(|N_{g(i)}^G| - |N_{t(i)}^T|)\right)^k; \tag{\stepcounter{equation}\theequation} \\
\Delta_n(s, m; k) &:= \frac{1}{n} \sum_{i \in N_n} \left(\sum_{h=s}^m (|N_{t(i)+h}^T| + |N_{t(i)-h}^T| + \mathbf{1}\{h = 0\}(|N_{g(i)}^G| - |N_{t(i)}^T|))\right)^k \tag{\stepcounter{equation}\theequation} \\
c_n(s, m; k) &:= \inf_{\alpha > 1} [\Delta_n(s, m; k\alpha)]^{1/\alpha} [\delta_n^\partial(s; \tfrac{\alpha}{\alpha-1})]^{1 - 1/\alpha}. \tag{\stepcounter{equation}\theequation}
\end{align*}
These objects can be interpreted in the following manner.
$\delta_n^\partial(s; k)$ tells us the average number of observations
that are $s$ periods away, while $\Delta_n(s, m; k)$ tells us the number
of observations between $m$ and $s$ periods away (inclusive), both
raised to the power of $k$. Finally, $c_n(s, m; k)$ gives us an upper
bound on the convolution of $\delta_n^\partial(s; k)$ and
$\Delta_n(s, m; k)$ when H\"older's inequality is applied, which is
interpretable as a measure of how quickly neighborhoods grow and
overlap. In the running example of a balanced panel,
$c_n(s, m; k) \asymp (m - s)^k G^k$.

The rest of this section characterizes $\psi$-dependence, which is the
sense in which observations are weakly dependent across clusters and
across time. Let $\mathcal{L}_{v,a}$ denote the collection of bounded
Lipschitz real functions on $\mathbb{R}^{v \times a}$:
\begin{equation}
\mathcal{L}_{v,a} = \{f: \mathbb{R}^{v \times a} \to \mathbb{R}: \|f\|_\infty < \infty, Lip(f) < \infty\},
\end{equation}
where $Lip(f)$ is the Lipschitz constant of $f$ and $\|\cdot\|_\infty$ is
the sup-norm. For any positive integers $a, b, s$, let
$\mathcal{Q}_n(a, b; s)$ denote the collection of all pairs of
observation sets of sizes $a$ and $b$ at distance at least $s$:
\begin{equation}
\mathcal{Q}_n(a, b; s) = \{(A, B) : A, B \subset N_n, |A| = a, |B| = b, \text{and } d_n(A, B) \geq s\},
\end{equation}
where $d_n(A, B) = \min_{i \in A} \min_{j \in B} d_n(i, j)$.

Inference is conditional on a $\sigma$-field $\mathcal C_n$ that holds fixed the information common to the array---for instance the cluster- and time-level common factors, or, in the design-based reading, the fixed potential outcomes. 
All expectations, variances, and covariances in this section and \Cref{sec:consistency_appendix} are taken conditional on $\mathcal C_n$, and the dependence coefficients $\theta_n$ below are $\mathcal C_n$-measurable.
Write $\mu_n:=\E[Y_n\mid\mathcal C_n]$ and $\Sigma_n:=\Var(Y_n\mid\mathcal C_n)$ for the conditional first and second moments; these are the $\mathcal C_n$-measurable analogues of the $\mu,\Sigma$ of \Cref{sec:framework}.
The dependence coefficients $\{\theta_{n,s}\}_{s\ge0}$ below (notation following KMS) are unrelated to the scalar variance target $\theta_n=\tr(W_n\Sigma)$ of \Cref{sec:framework}, which does not recur in this section.

\begin{definition}[$\psi$-dependence]
A triangular array $\{Y_{n,i}\}, n \geq 1, Y_{n,i} \in \mathbb{R}^v$ is
conditionally $\psi$-dependent if for each $n \in \mathbb{N}$, there
exist a measurable sequence $\theta_n = \{\theta_{n,s}\}_{s \geq 0},
\theta_{n,0} = 1$, and a collection of nonrandom functions
$(\psi_{a,b})_{a,b \in \mathbb{N}}, \psi_{a,b}: \mathcal{L}_{v,a} \times \mathcal{L}_{v,b} \to [0, \infty)$,
such that for all $(A, B) \in \mathcal{Q}_n(a, b; s)$ with $s > 0$ and
all $f \in \mathcal{L}_{v,a}$ and $g \in \mathcal{L}_{v,b}$,
\begin{equation}
|\Cov(f(Y_{n,A}), g(Y_{n,B}) \mid \mathcal C_n)| \leq \psi_{a,b}(f, g)\, \theta_{n,s}
\quad\text{a.s.}
\end{equation}
where $Y_{n,A} := (Y_{n,i})_{i \in A}$. We call the sequence $\theta_n$
the dependence coefficients of $\{Y_{n,i}\}$.
\end{definition}

This $\psi$-dependence setting generalizes the more familiar
strong-mixing processes used in the two-way cluster dependence panel
literature with serial correlation.

The setup here generalizes the CHS setting that requires a representation
where $Y_{gt} = f(\alpha_g, \gamma_t, \varepsilon_{gt})$, for iid
$\alpha_g, \varepsilon_{gt}$ and $\gamma_t$ strictly stationary, which
generalizes a standard Aldous--Hoover representation that additionally
requires $\gamma_t$ iid. The following example is a DGP that satisfies
the conditions of this paper and has homogeneous means but does not
permit the above representation.

\begin{example}
Suppose every firm $g$ is observed for $T$ time periods, each indexed by
$t$. Data is generated from the following process. There is an
innovation sequence $(\nu_{gt})_{t \in \mathbb{Z}}$ iid $N(0, 1)$ over
$g, t$. For some $|\rho| \in (0, 1)$, $D_{gt} = \rho D_{gt-1} + \nu_{gt}$,
$t \in \mathbb{Z}$, with $D_{gt}$ stationary and processes independent
across $g$. Hence, within a firm, $\{D_{gt}\}_t$ is an AR(1) Markov chain
with parameter $\rho$, and across firms, $\{D_{gt}\}_g$ are independent
copies of this AR(1) process.

Such a process does not permit the CHS representation. Suppose to a
contradiction that $D_{gt}$ admits a representation where
$D_{gt} = f(\alpha_g, \gamma_t, \varepsilon_{gt})$ for iid
$\alpha_g, \varepsilon_{gt}$, and strictly stationary $\gamma_t$. In the
DGP, $\gamma_t = 0$, because if $\gamma_t$ is stochastic, then there must
be dependence across $g$ units, which is ruled out in this DGP. With the
resulting representation $D_{gt} = f(\alpha_g, 0, \varepsilon_{gt})$, it
follows that $\Cov(D_{gt}, D_{gs_1}) = \Cov(D_{gt}, D_{gs_2})$ for
$t \neq s_1 \neq s_2 \neq t$. However, this is untrue in the DGP as
$\Cov(D_{g,t}, D_{g,t-1}) = \rho/(1-\rho^2) \neq \rho^2/(1-\rho^2) = \Cov(D_{g,t}, D_{g,t-2})$.
\end{example}

\subsection{Central limit theorem}
\label{sec:clt}

The results are stated in a setting where
$\sum_{i \in N_n} E[Y_{n,i}] = 0$, and we are interested in constructing
the variance of $\sum_{i \in N_n} Y_{n,i}$. With
$V_{\mathrm{true}} := Var(\sum_{i \in N_n} Y_{n,i})$, let
$\lambda_n := \lambda_{\min}(V_{\mathrm{true}})$ denote its smallest
eigenvalue.

\begin{assumption}
\label{asmp:psi_dependence}
The triangular array $\{Y_{n,i}\}$ is $\psi$-dependent with the
dependence coefficients $\{\theta_n\}$ satisfying:
(a) For some constant $C > 0$,
$\psi_{a,b}(f, g) \leq C \times ab(\|f\|_\infty + Lip(f))(\|g\|_\infty + Lip(g))$;
(b) $\sup_{n \geq 1} \max_{s \geq 1} \theta_{n,s} < \infty$;
(c) For some $p > 4$, $\sup_{n \geq 1} \max_{i \in N_n} \|Y_{n,i}\|_p < \infty$ a.s.
\end{assumption}

\begin{assumption}
\label{asmp:clt}
There exists a positive sequence $m_n \to \infty$ such that for
$k = 1, 2$, and the $p$ in \Cref{asmp:psi_dependence},
\begin{enumerate}[label=(\alph*)]
\item $\frac{n}{\lambda_n^{1+k/2}} \sum_{s \geq 0} c_n(s, m_n; k) \theta_{n,s}^{1 - (2+k)/p} \to 0$; and
\item $\frac{n^2 \theta_{n,m_n}^{1-(1/p)}}{\lambda_n^{1/2}} \to 0$.
\end{enumerate}
\end{assumption}

These conditions balance the dependence strength, neighborhood growth,
and variance growth. In an AR(1) process with positive coefficient
$\rho$, $\lambda_n \asymp T$, and both conditions are satisfied with
$m_n = T^{1/3}$ as $T \to \infty$. In the running balanced panel example
with common random effects within clusters and an AR(1) process in the
common time effects, $\lambda_n \asymp T^2 G$, and again
$m_n = T^{1/3}$ suffices.

\begin{theorem}
\label{thm:clt}
Suppose Assumptions \ref{asmp:psi_dependence} and \ref{asmp:clt} hold
a.s. Let $\nu$ be a nonstochastic vector with
$\dim(\nu) = \dim(Y_{n,i})$ and $\|\nu\| = 1$. Then, for
$S_n := \sum_{i \in N_n} \nu'(Y_{n,i} - E[Y_{n,i}])$ and
$\sigma_n^2 := Var(S_n)$,
\begin{equation}
\sup_{t \in \mathbb{R}} \left| \Pr\left\{\frac{S_n}{\sigma_n} \leq t\right\} - \Phi(t) \right| \xrightarrow{a.s.} 0,
\end{equation}
where $\Phi$ denotes the distribution function of $N(0, 1)$.
\end{theorem}

This result suffices for a uniform multivariate central limit theorem
due to the Cramer--Wold device.

\subsection{Consistency}
\label{sec:consistency}

By Corollary~\ref{cor:chs_optimal}, the optimal variance estimator is the exact
eigenvalue truncation $(W_n^{\mathrm{CHS}})^+ = W_n^{\mathrm{CHS}} + B^-\otimes Q_G$,
with $B^- = (C_M - J_T)_+$ and $Q_G = I_G - J_G/G$. In estimator form,
\begin{equation}
\widehat V_{\mathrm{ET}}^\star
:= \widehat V_{\mathrm{CHS}}
+ \sum_{g=1}^G \check Y_{n,(g,\cdot)}'\, B^-\, \check Y_{n,(g,\cdot)},
\qquad \check Y_{n,(g,t)} := Y_{n,(g,t)} - \bar Y_{\cdot t},
\label{eq:vet-estimator}
\end{equation}
where $\check Y_{n,(g,t)}$ is the score demeaned across groups within period $t$
and $\check Y_{n,(g,\cdot)} := (\check Y_{n,(g,t)})_{t=1}^T$. Equivalently
$\widehat V_{\mathrm{ET}}^\star = Y_n'(W_n^{\mathrm{CHS}} + B^-\otimes Q_G)Y_n$. The
correction is the within-group HAC of the cross-sectionally centered scores, using
the exact negative eigenmass $B^-$ of the indefinite temporal block $B = J_T - C_M$
as kernel; computing it is a single $T\times T$ eigendecomposition applied within
clusters. The corresponding estimand is
\begin{equation}
V_{\mathrm{ET}}^\star := \E[\widehat V_{\mathrm{ET}}^\star\mid\mathcal C_n]
= V_{\mathrm{CHS}}
+ \sum_{g=1}^G \E\big[\check Y_{n,(g,\cdot)}'\,B^-\,\check Y_{n,(g,\cdot)}\mid\mathcal C_n\big].
\label{eq:vet-estimand}
\end{equation}

Using $y_t := \sum_{i \in N_t^T} Y_{n,i}$, we compare the estimand to the kernel-adjusted variance
\begin{align}
V_{\mathrm{adj}} &:= \sum_{i \in N_n} \sum_{j \in N_{g(i)}^G} \Cov(Y_{n,i}, Y_{n,j})
 + \sum_{i \in N_n} \sum_{j \in N_{t(i)}^T} \Cov(Y_{n,i}, Y_{n,j})
 - \sum_{i \in N_n} \sum_{j \in N_{t(i),g(i)}^{T \cap G}} \Cov(Y_{n,i}, Y_{n,j}) \notag\\
 & \quad + \sum_{m=1}^M \sum_{t=1}^{T-m} \omega(m, M)\,\big[\Cov(y_t, y_{t+m}) + \Cov(y_{t+m}, y_t)\big],
\label{eq:vadj}
\end{align}
which will be shown to be asymptotically equivalent to $V_{\mathrm{true}}$. 
This $V_{\mathrm{adj}}$ is $\operatorname{tr}(W_n^{\mathrm{CHS}}\Sigma_n)$, the truncated kernel covariance: the variance part of the CHS plug-in ($\mathbb E[\widehat V_{\mathrm{CHS}}]=V_{\mathrm{adj}}+\mu_n'W_n^{\mathrm{CHS}}\mu_n$), and the quantity $\widehat V_{\mathrm{ET}}^\star$ targets.

We state the assumptions for the consistency theorem.
\begin{assumption}
\label{asmp:vcon_converge}
$\frac{n}{\lambda_n^2} \sum_{s \geq 0} c_n(s, M; 2)\, \theta_{n,s}^{1 - 4/p} \to 0$.
\end{assumption}
This assumption mimics Assumption 4.1(iii) in KMS. In the balanced panel the left-hand
side is $\asymp M^2/T^2$, so the condition reduces to $M=o(T)$; the choice
$M\asymp T^{1/3}$ used in the simulations and application satisfies it.

\begin{assumption}
\label{asmp:vadj_to_vtrue}
The following hold almost surely:
\begin{enumerate}[label=(\alph*)]
\item $\frac{ n}{\lambda_n} \sum_{m=1}^M |\omega(m, M) - 1| \delta_n^\partial(m; 1) \theta_{n,m}^{1 - 2/p} \to 0$;
\item $\frac{1}{\lambda_n} \sum_{m \geq 1} \sum_{t=1}^{T-m} \sum_{g=1}^G |N_{t,g}^{T \cap G}| |N_{t+m,g}^{T \cap G}| \theta_{n,m}^{1 - 2/p} \to 0$;
\item $\frac{1}{\lambda_n} \sum_{m=M+1}^{T-1} \sum_{t=1}^{T-m} |N_t^T| |N_{t+m}^T| \theta_{n,m}^{1 - 2/p} \to 0$.
\end{enumerate}
\end{assumption}
In the running AR(1) examples with $\rho \in (0, 1)$ and the triangular kernel, all
three conditions are implied by $\rho^m$ decay and the balanced panel structure
(where $|N_{t,g}^{T \cap G}|=1$).

\begin{theorem}[Consistency of the eigenvalue-truncation estimator]
\label{thm:panel_dep_consistent_main}
The following hold.
\begin{enumerate}[label=(\alph*)]
\item $V_{\mathrm{ET}}^\star - V_{\mathrm{adj}}$ is positive semidefinite.
\item Under Assumptions~\ref{asmp:psi_dependence} and~\ref{asmp:vcon_converge}, and
provided $\lambda_{\min}(V_{\mathrm{ET}}^\star)/\lambda_n \geq 1 + o(1)$ and
$Mn/\lambda_n\to0$,
\[
(V_{\mathrm{ET}}^\star)^{-1}\,\widehat V_{\mathrm{ET}}^\star \xrightarrow{p} I_v.
\]
\item Under Assumptions~\ref{asmp:psi_dependence} and~\ref{asmp:vadj_to_vtrue},
$V_{\mathrm{true}}^{-1} V_{\mathrm{adj}} \to I_v$.
\end{enumerate}
\end{theorem}

Observe that (c) implies $\lambda_{\min}(V_{\mathrm{adj}})/\lambda_n = 1 + o(1)$, so by
(a) $\lambda_{\min}(V_{\mathrm{ET}}^\star)/\lambda_n \geq 1 + o(1)$. The condition
$Mn/\lambda_n\to0$ in (b) is not an additional restriction: in the balanced panel
$n=GT$ and $\lambda_n\asymp T^2G$, so $Mn/\lambda_n = M/T$, which is exactly what
Assumption~\ref{asmp:vcon_converge} reduces to ($M=o(T)$, satisfied by
$M\asymp T^{1/3}$). Both auxiliary conditions in (b) are therefore automatically
satisfied under the assumptions.

\begin{corollary}[Asymptotically valid inference]
\label{cor:coverage}
Let $\nu\in\mathbb R^v$ with $\|\nu\|=1$, and suppose the conditions of
\Cref{thm:clt} and of \Cref{thm:panel_dep_consistent_main} hold. Then, for
$S_n=\sum_{i\in N_n}\nu'(Y_{n,i}-\E[Y_{n,i}])$ and any $\alpha\in(0,1)$,
\[
\limsup_{n\to\infty}\ \Pr\left\{
\frac{|S_n|}{\sqrt{\nu'\widehat V_{\mathrm{ET}}^\star\nu}}
> z_{1-\alpha/2}\right\}\ \le\ \alpha ,
\]
with equality if and only if $\nu'(V_{\mathrm{ET}}^\star-V_{\mathrm{true}})\nu=o\!\big(\nu'V_{\mathrm{true}}\nu\big)$.

\end{corollary}

\subsection{Extension to regression}

While the theoretical exposition focuses on means, the extension to the method of moments (which includes standard regressions) is straightforward. 
With a regression equation $Y_i = X_i' \beta + u_i$, and array $\{(X_i', u_i)\}$ satisfying $\psi$-dependence, as long as $\frac{1}{n} \sum_{i \in N_n} X_i X_i'$ has a minimum eigenvalue bounded away from zero, and $\{X_i u_i\}$ and $\{\mathrm{vec}(X_i X_i')\}$ satisfy the conditions of \Cref{thm:clt}, the OLS estimator $\widehat{\beta}$ is asymptotically normal and consistent. 
The variance of the coefficient estimator is computed by applying $\widehat{V}^\star_{\mathrm{ET}}$ to the estimated score $X_i \widehat{u}_i$ in place of $Y_{n,i}$. 
Standard OLS arguments then apply when limit theorems hold. 
Writing $X_i\widehat u_i = X_iu_i - X_iX_i'(\widehat\beta-\beta)$, the feasible estimator differs from its infeasible counterpart $\sum_{i,j}(A_n^\star)_{ij}(X_iu_i)(X_ju_j)'$ only through terms linear and quadratic in $\widehat\beta-\beta$, which the conditions of \Cref{thm:clt} render negligible provided $\max_i\|X_i\|^2\,\|\widehat\beta-\beta\|$ is negligible at the CLT rate.

In the numerical results that follow, we additionally augment the methods with ``C2NW", which is a natural extension of CGM2 to cross-cluster serial correlation. This method is conservative, but is not necessarily optimal --- its details are relegated to \Cref{sec:panel_dep_appendix}.

\subsection{Simulation}
\label{sec:simulation}

The simulations show how the variance estimators that are robust to
heterogeneous means compare with existing methods, and isolate the
regime in which the optimal correction delivers a concrete advantage.
Data are generated from the linear model
\begin{equation}
Y_{gt} = \beta_0 + (\beta_1 + \eta_{gt})\,X_{gt} + \delta_{gt} + U_{gt},
\label{eq:sim-dgp}
\end{equation}
where the regressor and disturbance follow the CHS variance-component
structure,
\begin{align*}
X_{gt} &= w_\alpha \alpha^x_g + w_\gamma \gamma_t^x + w_\varepsilon \varepsilon^x_{gt}, \\
U_{gt} &= w_\alpha \alpha^u_g + w_\gamma \gamma_t^u + w_\varepsilon \varepsilon^u_{gt},
\tag{\stepcounter{equation}\theequation}
\end{align*}
with $(\beta_0, \beta_1) = (0.1, 0.1)$ and $(w_\alpha, w_\gamma,
w_\varepsilon) = (0.15, 0.20, 0.15)$. The latent components
$(\alpha^x_g, \alpha^u_g, \varepsilon^x_{gt}, \varepsilon^u_{gt})$ are
mutually independent $N(0,1)$, and the common time effects $(\gamma^x_t,
\gamma^u_t)$ follow AR(1) processes with coefficient $\rho$.

The two heterogeneity terms are rank-one cluster-by-time interactions
\[
\eta_{gt} = \kappa\, c^\eta_g\, d^\eta_t,
\qquad
\delta_{gt} = \delta\, c^\delta_g\, d^\delta_t,
\]
where the loading vectors $c^\bullet_g$ and $d^\bullet_t$ are independent, mean-zero draws (the $d^\bullet_t$ are serially uncorrelated). 
Centering both factors places each interaction in the two-way demeaned subspace $\mathcal{S}_{00}$---precisely the subspace on which the two-way clustering weight matrix carries its negative eigenvalue (Section~\ref{sec:examples-and-optimal}). 
The term $\eta_{gt}$ is a mean-zero random \emph{slope}: because it is mean-zero and independent of $X$, the OLS estimand for $\beta_1$ remains exactly $0.1$, so the test is correctly centered; but the score $X_{gt} u_{gt}$ inherits the component $\eta_{gt} X_{gt}^2$, which has a heterogeneous \emph{mean} on $\mathcal{S}_{00}$. 
The term $\delta_{gt}$ is an additive outcome-level component omitted from the regression, contributing heterogeneous score \emph{variance} on the same subspace. 
We fix $\kappa = 0.8$ and report results with the level heterogeneity switched off ($\delta = 0$) and on ($\delta = 0.8$). We hold $G = 50$, and use $T=120$ for most designs, comparable to the dimensions of the empirical application in Section~\ref{sec:empirical}, and vary $\rho \in \{0.25, 0.50, 0.75\}$.

In each of $1{,}000$ draws we regress $Y_{gt}$ on $X_{gt}$ and a constant, compute standard errors by each method, and test $H_0: \beta_1 = 0.1$ at the $5\%$ level.
Table~\ref{tab:sim_results} reports the empirical rejection rates. 
Across every design the Monte Carlo mean of $\widehat\beta_1$ remains within sampling error of $0.1$, confirming the test is correctly centered and that the reported rates are size, not bias.

\begin{table}[H]
\centering
\caption{Rejection rates of nominal $5\%$ tests of
$H_0:\beta_1 = 0.1$. $G = 50$; $1{,}000$ Monte Carlo draws. ``Level Het.''
indicates whether the additive cluster-by-time level heterogeneity is
present ($\delta = 0.8$) or absent ($\delta = 0$); the random-slope
heterogeneity $\kappa = 0.8$ is present throughout.}
\label{tab:sim_results}
\begin{tabular}[t]{ccccccccccc}
\toprule
row & T & rho & Level Het. & EHW & CRg & CRt & CGM & CHS & C2NW & ET\\
\midrule
(I) & 120 & 0.25 & No & 0.477 & 0.220 & 0.157 & 0.075 & 0.071 & 0.010 & 0.055\\
(II) & 120 & 0.50 & No & 0.521 & 0.254 & 0.204 & 0.121 & 0.091 & 0.016 & 0.074\\
(III) & 120 & 0.75 & No & 0.631 & 0.364 & 0.343 & 0.215 & 0.120 & 0.041 & 0.107\\
\addlinespace
(IV) & 120 & 0.25 & Yes & 0.201 & 0.130 & 0.115 & 0.078 & 0.075 & 0.000 & 0.026\\
(V) & 120 & 0.50 & Yes & 0.245 & 0.173 & 0.140 & 0.098 & 0.080 & 0.001 & 0.036\\
(VI) & 120 & 0.75 & Yes & 0.364 & 0.274 & 0.267 & 0.200 & 0.098 & 0.012 & 0.068\\
\addlinespace
(VII) & 1000 & 0.75 & No & 0.734 & 0.133 & 0.510 & 0.104 & 0.058 & 0.031 & 0.057\\
(VIII) & 1000 & 0.75 & Yes & 0.551 & 0.137 & 0.457 & 0.125 & 0.064 & 0.030 & 0.056\\
\bottomrule
\end{tabular}

\justifying \small
EHW denotes Eicker--Huber--White; CRg and CRt denote one-way
cluster-robust within $g$ and within $t$; CGM denotes
Cameron--Gelbach--Miller; CHS denotes Chiang--Hansen--Sasaki (as written in their paper, without subtracting within-cluster lag); C2NW denotes
the conservative estimator $\widehat{V}_{\mathrm{C2NW}}$; ET denotes the
eigenvalue-truncation estimator $\widehat{V}^\star_{\mathrm{ET}}$.
\end{table}

Three patterns emerge. 
First, the estimators that ignore dependence (EHW, CRg, CRt) over-reject grossly throughout, confirming that some form of dependence robustness is essential. 
Second, among the dependence-robust estimators, CGM and CHS over-reject substantially when mean heterogeneity is present, and the distortion grows with the persistence $\rho$.
ET controls size far better in these designs, because its correction targets exactly the directions on which the heterogeneity loads. 
The conservative estimator C2NW controls size in all designs, at the cost of being noticeably conservative (rejection rates well below nominal).

Third, ET uniformly has lower rejection rates than CHS, including in the designs without level heterogeneity ($\delta = 0$, rows I--III). 
This reflects the within-cluster correction embedded in $\widehat{V}^\star_{\mathrm{ET}}$. 
The gap in size distortion is modest without level heterogeneity and widens once heterogeneity is present.

There are two qualifications. The cluster-by-time heterogeneity here is serially \emph{uncorrelated}; this is what prevents the CHS HAC lag estimator from absorbing it and exposes the bias in the CHS estimand. 
Under serially correlated heterogeneity the CHS lag term partially captures the structure and CHS is correspondingly more robust in finite samples. 
While ET controls size far better than CGM and CHS under heterogeneity, there is some over-rejection in finite-sample designs, but this over-rejection disappears when we increase $T$ to $1000$, even for the hardest case with high $\rho$. 

\subsection{Empirical Application}
\label{sec:empirical}

The new variance estimator is applied to a panel of industry portfolios across time. 
The Fama--French three-factor model \citep{fama1993common}
can be written as:
\begin{equation}
R_{gt} - R_{ft} = \beta_1 (R_{Mt} - R_{ft}) + \beta_2 SMB_{gt} + \beta_3 HML_{gt} + e_{gt},
\end{equation}
where $R_{gt} - R_{ft}$ is the excess return as $R_{gt}$ denotes the
total return of portfolio $g$ in month $t$, and $R_{ft}$ denotes the
risk-free rate.

I use the industry portfolio dataset from CHS containing 44 industry
portfolios across 119 months. Following CHS, I use a within-transformation
for both $Y$ and $X = (R_M - R_f, SMB, HML)$ to remove additive fixed
effects. Inference on coefficients relies on a CLT applied to
$\sum_{g,t} \widetilde{X}_{gt} \widetilde{e}_{gt}$.

Allowing for heterogeneous means is natural here when the three-factor model is a linear approximation of the return process, and inference conditions on the portfolios and sample months observed. 
Under correct specification the score $\widetilde X_{gt}\widetilde e_{gt}$ would be conditionally mean-zero, and hence no mean heterogeneity.
But suppose a priced factor is omitted---momentum, profitability, investment, or another of the many documented since \citet{fama1993common}. 
Conditional on the observed portfolios and months, the omitted portfolio-by-time specific factor is not a random shock but a fixed feature of cell $(g,t)$: it enters the \emph{mean} of the score $\widetilde X_{gt}\widetilde e_{gt}$. 
We reserve variance for what remains genuinely stochastic given the conditioning---the idiosyncratic innovations---and treat the systematic, omitted-factor contribution as a  heterogeneous conditional mean rather than as excess score variance. 
The within-transformation we apply, following CHS, removes any additive portfolio-level and time-level component of this mean, so a constant industry alpha or a common-to-all-portfolios mispricing in a given month is swept out.
What survives is the \emph{non-separable} interaction between portfolio and time which is precisely the two-way demeaned component on which the two-way clustering weight matrix carries its negative eigenvalue (\Cref{sec:examples-and-optimal}). 
The researcher need not estimate the omitted factor: the role of the eigenvalue truncation is instead to modify the estimand so that inference remains conservative against whatever non-separable mean heterogeneity the omitted structure induces.

\begin{table}[H]
\centering
\caption{Standard Errors for Various Methods: 44 Industry Portfolios, 119 Periods}
\label{tab:emp_results}

\begin{tabular}{lrrrrrrrr}
\toprule
  & coef & EHW & CRg & CRt & CGM & CHS & C2NW & ET\\
\midrule
MKT & 0.959 & 0.022 & 0.055 & 0.031 & 0.059 & 0.059 & 0.076 & 0.062\\
SMB & 0.076 & 0.029 & 0.035 & 0.041 & 0.045 & 0.053 & 0.083 & 0.059\\
HML & 0.358 & 0.030 & 0.066 & 0.049 & 0.076 & 0.080 & 0.109 & 0.084\\
\bottomrule
\end{tabular}

\justifying \small
Estimates of the factor coefficients with various standard error
estimates. EHW denotes Eicker--Huber--White; CRg denotes one-way
cluster-robust within $g$; CRt denotes one-way cluster-robust within
$t$; CGM denotes Cameron--Gelbach--Miller; CHS denotes
Chiang--Hansen--Sasaki; C2NW denotes the conservative estimator
$\widehat{V}_{\mathrm{C2NW}}$; ET denotes the optimal estimator
$\widehat{V}^\star_{\mathrm{ET}}$. 
\end{table}

From Table~\ref{tab:emp_results}, standard errors estimated using ET are higher than the standard errors estimated using CHS.
However, the standard errors are not unreasonably large: we retain the statistical significance of the coefficient on HML at the 5\% level, while the statistical significance of SMB can now be called into question, as CHS already show. 
A further takeaway is that cross-cluster serial correlation is empirically important: if serial correlation were negligible, we would expect the CHS standard errors to be similar to CGM, which is not the case for two of the coefficients. 
We verify that ET meaningfully decreases the standard errors relative to the more conservative C2NW, as expected due to how ET is optimal.

\section{Conclusion}
\label{sec:conclusion}

This paper developed a unified theory of dependence-robust variance estimation under heterogeneous means. 
Every quadratic estimator $Y_n'A_nY_n$ has expectation $\tr(A_n\Sigma)+\mu'A_n\mu$, so its behavior is governed by two structural commitments---a subspace $\mathcal M_n$ of admissible means and a closed convex cone $\mathcal S_n$ of admissible covariances. 
Cone duality turns this decomposition into an exact conservativeness criterion (\Cref{thm:conservative}), and the eigenvalue truncation $W_n^+$ minimizes the adjustment to the plug-in and, in the aligned regime, the mean-squared error and estimand level (\Cref{thm:eig-truncation}). 
The same construction doubles as a diagnostic: one-way cluster-robust estimators, and Bartlett-kernel HAC emerge as optimal, whereas two-way clustering is anticonservative under heterogeneous means and is repaired by an explicit, minimal correction. 
Specializing the theory resolves an open problem in the panel literature--- a variance estimator simultaneously robust to heterogeneous means and cross-cluster serial correlation, with consistency and conservativeness established under $\psi$-dependence and borne out in simulations and a Fama--French application.

\appendix
\crefalias{section}{appendix}
\crefalias{subsection}{appendix}

\section{Supplementary Results}
\label{sec:panel_dep_appendix}

\subsection{Details on Optimization Problems}

In the minimal correction problem, the constraint $\|\mu\| = 1$ excludes $\mu =0$. 
When $\mathcal{M}_n = \{0\}$---the mean is known exactly, so that the demeaned mean vanishes---the supremum in~\eqref{eq:risk-bias} is over the empty set; we adopt the convention $\mathcal{R}_n(A_n) := 0$, reflecting that (C2) is vacuous and no mean-induced correction is warranted, so the plug-in $A_n = W_n$ is feasible and optimal.
The criterion reads $A_n$ only through its $\mathcal{M}_n$-block $P_{\mathcal{M}_n}A_nP_{\mathcal{M}_n}$ and is silent about the off-$\mathcal{M}_n$ block. 

Under finite fourth moments, expanding the squared error gives
\begin{equation}
\mathrm{MSE}_n(A_n; P_n)
=
\bigl[\tr((A_n - W_n)\Sigma) + \mu' A_n \mu\bigr]^2
+ 2\tr((A_n\Sigma)^2)
+ 4\mu' A_n \Sigma A_n \mu
+ \mathcal{T}_3(A_n; P_n)
+ \mathcal{K}_4(A_n; P_n),
\label{eq:mse-expansion}
\end{equation}
where, writing $Z = Y_n - \mu$, $m^{(3)}_{ijk}(P_n) = \E_{P_n}[Z_i Z_j Z_k]$, and
$\kappa_{ijkl}(P_n) = \E_{P_n}[Z_iZ_jZ_kZ_l] - \E[Z_iZ_j]\E[Z_kZ_l] - \E[Z_iZ_k]\E[Z_jZ_l] - \E[Z_iZ_l]\E[Z_jZ_k]$,
\[
\mathcal{T}_3(A_n; P_n) := 4 \sum_{i,j,k} (A_n)_{ij} (A_n \mu)_k\, m^{(3)}_{ijk}(P_n),
\qquad
\mathcal{K}_4(A_n; P_n) := \sum_{i,j,k,l} (A_n)_{ij}(A_n)_{kl}\, \kappa_{ijkl}(P_n).
\]

Under Assumption~\ref{asmp:no-skew}, MSE reduces to
\begin{equation}
\mathrm{MSE}_n(A_n; \mu, \Sigma)
= \bigl[\tr((A_n - W_n)\Sigma) + \mu' A_n \mu\bigr]^2
+ 2\tr((A_n\Sigma)^2)
+ 4\mu' A_n \Sigma A_n \mu,
\label{eq:mse-clean}
\end{equation}
depending on $P_n$ only through $(\mu, \Sigma)$.

\begin{remark}[Both restrictions in \Cref{thm:eig-truncation}(iii) are essential]
\label{rem:mse-scope}
The optimality in Theorem~\ref{thm:eig-truncation}(iii) holds over the
\emph{spectral} class under a \emph{commuting} cone, and fails if either
restriction is dropped. $A_n^\star$ is in particular \emph{not} globally
MSE-minimal, and feasible Loewner-refinements $A_n\succeq A_n^\star$ can
strictly improve on it.

\emph{The cone must commute with $W_n$.} Take $\mathcal M_n=\R^2$,
$\mathcal S_n=\mathrm{Sym}_2^+$, $Y_n$ Gaussian, $W_n=\diag(1,-1)$, so
$A_n^\star=W_n^+=\diag(1,0)$. The feasible correction
$A_n=\bigl(\begin{smallmatrix}2&1\\1&1/2\end{smallmatrix}\bigr)$ satisfies
$A_n\succeq0$, $A_n-W_n\succeq0$, and at $\mu=0$,
$\Sigma=\bigl(\begin{smallmatrix}1&-1\\-1&1\end{smallmatrix}\bigr)$ attains
$\mathrm{MSE}_n(A_n)=3/4<3=\mathrm{MSE}_n(A_n^\star)$: a correction not sharing
the eigenstructure of $ W_n$ places less weight where $\Sigma$
concentrates. Here $\Sigma$ does not commute with $W_n$.

\emph{The correction must be spectral, even when the cone commutes.} Take
$\mathcal M_n=\R^2$, $Y_n$ Gaussian, $\mathcal S_n=\mathcal S_2^{\mathrm{diag}}$
(so condition (c)'s commuting holds), and
\[
W_n=\begin{pmatrix}1&0\\0&0\end{pmatrix},\qquad A_n^\star=W_n^+=W_n .
\]
Fix $\mu=(1,1)'$, $\Sigma=\diag(1,\sigma)$ with $0<\sigma<1/9$, and move into the
\emph{non-spectral} feasible direction $D=vv'$, $v=(1,-3)'$, i.e.
$A_n=A_n^\star+tD$ for small $t>0$. Then $A_n-A_n^\star=tD\succeq0$ and
$\diag(A_n-W_n)=t(1,9)\ge0$, so $A_n$ is feasible and Loewner-dominates
$A_n^\star$. Yet, using
$\Cov(Y_n'A_n^\star Y_n,Y_n'DY_n)=2\tr(A_n^\star\Sigma D\Sigma)
+4\mu'A_n^\star\Sigma D\mu=2-8=-6$ and
$\E[Y_n'A_n^\star Y_n]-\theta_n=\mu'A_n^\star\mu=1$,
$\E[Y_n'DY_n]=\tr(D\Sigma)+\mu'D\mu=5+9\sigma$,
\[
\frac{d}{dt}\,\mathrm{MSE}_n(A_n^\star+tD;\mu,\Sigma)\Big|_{t=0}
=2\,\E\!\big[(Y_n'A_n^\star Y_n-\theta_n)\,Y_n'DY_n\big]
=2\big(-6+(5+9\sigma)\big)=2(9\sigma-1)<0 .
\]
So $\mathrm{MSE}_n(A_n)<\mathrm{MSE}_n(A_n^\star)$ for small $t>0$: a feasible,
conservative, Loewner-dominating refinement strictly beats the truncation. The
obstruction is the off-eigenbasis tilt of $D$, which lowers the
mean--variance cross term $4\mu'A_n\Sigma A_n\mu$ faster than it raises bias;
confining the search to $\mathcal A_n^{\mathrm{sp}}$ removes exactly this freedom.
\end{remark}

\subsection{Examples for Anticonservativeness}

We provide two illustrative examples. The first example shows how the plug-in estimator is anticonservative with m-dependence. The second example shows that the CHS estimand is anticonservative under heterogeneous means. 

\begin{example}
Suppose we have an $m$-dependent series, where $m=1$, and $3$ time periods. 
This example can be extended to the $T\rightarrow\infty$ case for a sequence of heterogeneous means. 
The variance for the sum of a scalar random variable $y_{t}$ is:
\begin{equation}
\Var\left(\sum_{t=1}^{3}y_{t}\right)  =\sum_{t=1}^{3}\Var\left(y_{t}\right)+2\sum_{t=1}^{2}\Cov\left(y_{t},y_{t+1}\right).
\end{equation}

When using a standard heteroskedasticity and autocorrelation robust (HAR) variance estimator, we target $\sum_{t=1}^{3}E\left[y_{t}^{2}\right]+2\sum_{t=1}^{2}E\left[y_{t}y_{t+1}\right]$, but
\begin{align*}
D_{1} & :=\sum_{t=1}^{3}E\left[y_{t}^{2}\right]+2\sum_{t=1}^{2}E\left[y_{t}y_{t+1}\right]-\Var\left(\sum_{t=1}^{3}y_{t}\right)\\
 & =\sum_{t=1}^{3}E\left[y_{t}\right]^{2}+2\sum_{t=1}^{2}E\left[y_{t}\right]E\left[y_{t+1}\right]. \tag{\stepcounter{equation}\theequation}
\end{align*}

The expression $D_{1}$ above can be negative while having $\sum_{t=1}^3 E\left[y_{t}\right]=0$ when $E[y_1] = 0.5, E[y_2] = -1, E[y_3]= 0.5$. 
Then, $\sum_{t=1}^{3}E\left[y_{t}\right]^{2}+2\sum_{t=1}^{2}E\left[y_{t}\right]E\left[y_{t+1}\right]=1.5-2=-0.5$.
To restore validity, we may add (a scaled version of) $\sum_{t=1}^{3}y_{t}^{2}$ to the variance estimator, so that we target $2\sum_{t=1}^{3}E\left[y_{t}^{2}\right]+2\sum_{t=1}^{2}E\left[y_{t}y_{t+1}\right]$ instead. 
Using this expression,
\begin{align*}
D_{2} & :=2\sum_{t=1}^{3}E\left[y_{t}^{2}\right]+2\sum_{t=1}^{2}E\left[y_{t}y_{t+1}\right]-\Var\left(\sum_{t=1}^{3}y_{t}\right)\\
 & =2\sum_{t=1}^{3}E\left[y_{t}\right]^{2}+2\sum_{t=1}^{2}E\left[y_{t}\right]E\left[y_{t+1}\right]+\sum_{t=1}^{3}\Var\left(y_{t}\right)\\
 & \geq\sum_{t=1}^{2}E\left[y_{t}\right]^{2}+2\sum_{t=1}^{2}E\left[y_{t}\right]E\left[y_{t+1}\right]+\sum_{t=1}^{2}E\left[y_{t+1}\right]^{2}+\sum_{t=1}^{3}\Var\left(y_{t}\right)\geq 0, \tag{\stepcounter{equation}\theequation}
\end{align*}
so the variance estimand is conservative and robust to mean heterogeneity.
When the means are homogeneous, we would overestimate the true variance by $D_{2}=\sum_{t=1}^{3}\Var\left(y_{t}\right)$. 
\end{example}

\begin{example}
Suppose we have one observation in each $\left(g,t\right)$ intersection, and we have $m=1$ dependence. 
We can compare the oracle CHS variance estimand (where $\omega(m,M)=1$ for all $m\leq 1$ and 0 otherwise) with the true variance.
As before, this example can be extended to $T \rightarrow \infty$ and $G \rightarrow \infty$. 
The scalar $E\left[Y_{n,i}\right]$ values are in \Cref{tab:chs_counterexample}.

\begin{table}[H]
\caption{Values of scalar $E[Y_{n,i}]$}
\label{tab:chs_counterexample}
\centering
\begin{tabular}{|c|c|c|c|c|}
\hline 
 & $t=1$ & $t=2$ & $t=3$ & $t=4$ \tabularnewline
\hline 
$g=1$ & $-1$ & $-1$ & $1$ & $1$ \tabularnewline
\hline 
$g=2$ & $1$ & $1$ & $-1$ & $-1$ \tabularnewline
\hline 
\end{tabular}
\end{table}

Then, observe that $\sum_{i \in N_n}\sum_{j\in N_{g(i)}^{G}}E\left[Y_{n,i}\right]E\left[Y_{n,j}\right]=0$, $\sum_{i \in N_n}\sum_{j\in N_{t(i)}^{T}}E\left[Y_{n,i}\right]E\left[Y_{n,j}\right]=0$, and $\sum_{i \in N_n}\sum_{j\in N_{t(i),g(i)}^{T\cap G}}E\left[Y_{n,i}\right]E\left[Y_{n,j}\right]=8$.
In the serially correlated components, setting $m=1$, $\sum_{t=1}^{T-m}E\left[y_{t}\right]E\left[y_{t+m}^{\prime}\right]=0$ and $\sum_{t=1}^{T-m}\sum_{g=1}^{G}\sum_{i\in N_{t,g}^{T\cap G}}\sum_{j\in N_{t+m,g}^{T\cap G}}E\left[Y_{n,i}\right]E\left[Y_{n,j}\right]=2$.
Then, on aggregate, 
\begin{align*}
& \sum_{i \in N_n}\sum_{j\in N_{g(i)}^{G}}E\left[Y_{n,i}\right]E\left[Y_{n,j}\right] + \sum_{i \in N_n}\sum_{j\in N_{t(i)}^{T}}E\left[Y_{n,i}\right]E\left[Y_{n,j}\right] - 
\sum_{i \in N_n}\sum_{j\in N_{t(i),g(i)}^{T\cap G}}E\left[Y_{n,i}\right]E\left[Y_{n,j}\right] \\
& +2\sum_{t=1}^{T-m}E\left[y_{t}\right]E\left[y_{t+m}\right] -
2\sum_{t=1}^{T-m}\sum_{g=1}^{G}\sum_{i\in N_{t,g}^{T\cap G}}\sum_{j\in N_{t+m,g}^{T\cap G}}E\left[Y_{n,i}\right]E\left[Y_{n,j}\right] = -12<0, \tag{\stepcounter{equation}\theequation}
\end{align*}
so the standard CHS estimand is anticonservative. 
\end{example}

\subsection{Variance estimators}
\label{sec:two-estimators}

We now define both variance estimators and their target estimands.
The CHS variance estimator (excluding their EVC operation) is:
\begin{align}
\widehat{V}_{\mathrm{CHS}}
&= \sum_{i \in N_n} \sum_{j \in N_{g(i)}^G} Y_{n,i} Y_{n,j}'
 + \sum_{i \in N_n} \sum_{j \in N_{t(i)}^T} Y_{n,i} Y_{n,j}'
 - \sum_{i \in N_n} \sum_{j \in N_{t(i),g(i)}^{T \cap G}}
   Y_{n,i} Y_{n,j}' \notag \\
&\quad + \sum_{m=1}^M \omega(m,M)
\left(
\sum_{t=1}^{T-m} y_t y_{t+m}'
+ \sum_{t=1}^{T-m} y_{t+m} y_t'
\right) \notag \\
&\quad - \sum_{m=1}^M \omega(m,M)
\left(
\sum_{t=1}^{T-m} \sum_{g=1}^G
\sum_{i \in N_{t,g}^{T \cap G}}
\sum_{j \in N_{t+m,g}^{T \cap G}}
Y_{n,i}Y_{n,j}'
\right.  \left.
+ \sum_{t=1}^{T-m} \sum_{g=1}^G
\sum_{i \in N_{t,g}^{T \cap G}}
\sum_{j \in N_{t+m,g}^{T \cap G}}
Y_{n,j}Y_{n,i}'
\right).
\label{eq:Vchs}
\end{align}
The first line is the CGM estimator, the second accounts for serial
correlation across clusters, and the third adjusts for double counting.
The CHS estimand $V_{\mathrm{CHS}}$ replaces each empirical product
with its expectation.

\begin{lemma}[Matrix form of the CHS estimator]
\label{lem:chs-matrix-form}
In the balanced panel with one observation per cell, $i=(g,t)$ and
$\mathbb{R}^n\cong\mathbb{R}^T\otimes\mathbb{R}^G$, the CHS variance estimator
\eqref{eq:Vchs} admits the quadratic-form representation
\[
\widehat V_{\mathrm{CHS}} = Y_n'W_n^{\mathrm{CHS}}Y_n,
\qquad
W_n^{\mathrm{CHS}} = J_T\otimes I_G + C_M\otimes(J_G - I_G).
\]
\end{lemma}

A simpler conservative estimator extends \citet{davezies2021empirical}'s CGM2 construction to the serially correlated setting. 
CGM2 adds $W_n^T$ to the CGM weight matrix to restore uniform conservativeness in the standard two-way setting (Section~\ref{sec:twoway_optimal}); its analog here adds a tractable PSD term $2 \sum_t y_t y_t'$ in the kernel-weighted block:
\begin{align*}
\widehat{V}_{\mathrm{C2NW}} &:= \sum_{i \in N_n} \sum_{j \in N_{g(i)}^G} Y_{n,i} Y_{n,j}' + \sum_{i \in N_n} \sum_{j \in N_{t(i)}^T} Y_{n,i} Y_{n,j}' \\
 & \quad + \sum_{m=1}^M \omega(m, M) \left(\sum_{t=1}^{T-m} y_t y_{t+m}' + \sum_{t=1}^{T-m} y_{t+m} y_t' + 2\sum_{t=1}^T y_t y_t'\right). \tag{\stepcounter{equation}\theequation}
\end{align*}
$V_{\mathrm{C2NW}}$ is defined analogously by replacing empirical
products with expectations. The estimator omits the
$\sum_{t,g,i,j}$ double-counting line of $\widehat{V}_{\mathrm{CHS}}$,
matching CHS's own implementation choice.

By construction, $\widehat{V}_{\mathrm{C2NW}}$ adds weight on the
group-symmetric subspace ($v \otimes \mathbf{1}_G$ directions), while
$\widehat{V}^\star_{\mathrm{ET}}$ adds weight precisely on the group-demeaned
subspace where the negative eigenvalues live. Both are uniformly
conservative, but the optimal correction targets the pathology
directly while the conservative correction overshoots.

For implementation, $M$ is selected using the bandwidth selection
procedure of \citet{andrews1991heteroskedasticity} and $\omega(\cdot)$ is
the triangular kernel.
This choice follows CHS and is used here for comparability, but we note that the
\citet{andrews1991heteroskedasticity} rule is derived under stationarity for a
single time series and its optimality has not been established for the
$\psi$-dependent two-way array considered here; \Cref{asmp:vcon_converge} requires
only $M=o(T)$, which the resulting bandwidths satisfy.

\subsection{Consistency Details}
\label{sec:consistency_appendix}

\begin{lemma}[Quadratic-form law of large numbers]
\label{lem:quadratic_LLN}
Let $\{Y_{n,i}\}$ be a triangular array of scalar random variables that is conditionally $\psi$-dependent given $\{\mathcal{C}_n\}$, with dependence coefficients $\{\theta_{n,s}\}$ satisfying Assumption~\ref{asmp:psi_dependence}(b). 
Let $A_n \in \mathrm{Sym}_n$ be a (possibly random but $\mathcal{C}_n$-measurable) symmetric weight matrix, and define
\[
\hat{V}_{A,n} := Y_n' A_n Y_n, \qquad V_{A,n} := \mathbb{E}[Y_n' A_n Y_n \mid \mathcal{C}_n].
\]
Suppose the following primitive conditions hold for some $p > 4$:
\begin{enumerate}
\item[(i)] (Moment bound) $\sup_{n \geq 1} \max_{i \in N_n} \|Y_{n,i}\|_{\mathcal{C}_n, p} < \infty$ a.s.
\item[(ii)] (Weight bound) There is a nonnegative sequence $\{\bar a_n\}$ with
$\max_{i,j \in N_n} |(A_n)_{ij}| \leq \bar{a}_n$ a.s.\ for every $n$.
\item[(iii)] (Support bandwidth) There exists a sequence $b_n \geq 0$ such that 
$(A_n)_{ij} = 0$ whenever $d_n(i, j) > b_n$.
\item[(iv)] (Variance summability)
\[
\bar{a}_n^2 \cdot n \sum_{s \geq 0} c_n(s, b_n; 2) \, \theta_{n,s}^{1 - 4/p} \to 0 \quad \text{a.s.}
\]
\end{enumerate}
Then
\[
\hat{V}_{A,n} - V_{A,n} \xrightarrow{p} 0. 
\]
\end{lemma}

The lemma is applied to weight matrices that have already been normalized by
$\lambda_n$: for an unnormalized $A_n$ with $\max_{i,j}|(A_n)_{ij}|\le a_n$, taking
$A_n/\lambda_n$ in the lemma gives $\bar a_n = a_n/\lambda_n$ and recovers
$\lambda_n^{-1}(\hat V_{A,n}-V_{A,n})\xrightarrow{p}0$ under
$\tfrac{a_n^2}{\lambda_n^2}\,n\sum_{s\ge0}c_n(s,b_n;2)\theta_{n,s}^{1-4/p}\to0$.
Condition (ii) alone does not deliver (iv): the sum
$n\sum_{s\ge0}c_n(s,b_n;2)\theta_{n,s}^{1-4/p}$ diverges, so the rate at which
$\bar a_n$ vanishes is what does the work.

\begin{remark}
Conditions (i), (ii), and (iv) are direct adaptations of Assumption~4.1 of KMS to
matrix-form weights. Condition (iii) holds whenever $A_n^\star$ inherits the support
pattern of the underlying plug-in $W_n$, as for the block and banded cones of
\Cref{sec:framework}. It fails for the exact truncation of \Cref{cor:chs_optimal}, whose
correction $B^-\otimes Q_G$ is globally supported in time; that case is handled in
\Cref{prop:exact-opt-consistent} below by applying the lemma to the band-limited proxy and
bounding the difference in operator norm.
\end{remark}

We establish consistency of the eigenvalue-truncation estimator
$\widehat V_{\mathrm{ET}}^\star$ of~\eqref{eq:vet-estimator} and of the conservative
estimator $\widehat V_{\mathrm{C2NW}}$ of Appendix~\ref{sec:two-estimators}. The
argument has three layers. First, two band-limited quadratic forms---the conservative
$\widehat V_{\mathrm{C2NW}}$ and a proxy $\widehat V_{\mathrm{band}}$ for the
truncation---obey a law of large numbers directly through
Lemma~\ref{lem:quadratic_LLN}; the C2NW branch additionally invokes
\Cref{asmp:vcon_converge_c2nw}, since its within-period block is not entrywise bounded.  Second, the exact truncation
$\widehat V_{\mathrm{ET}}^\star$, which is not band-limited, inherits consistency from
the proxy $\widehat V_{\mathrm{band}}$ because the two differ by an operator of vanishing relative size. Third, the
common target $V_{\mathrm{adj}}$ is asymptotically equivalent to $V_{\mathrm{true}}$.

Throughout, $\widehat V_{\mathrm{band}} := Y_n'A_n^{\mathrm{band}}Y_n$ for
\[
A^{\mathrm{band}}_n := W_n^{\mathrm{CHS}} + C_M\otimes Q_G
= J_T\otimes I_G + \tfrac{G-1}{G}\,C_M\otimes J_G,
\qquad V_{\mathrm{band}} := \E[\widehat V_{\mathrm{band}}\mid\mathcal C_n],
\]
the band-limited proxy; it plays no role beyond the consistency proof.

\begin{proposition}[Conservativeness]
\label{prop:both-conservative}
$V_{\mathrm{ET}}^\star - V_{\mathrm{adj}}$ and $V_{\mathrm{C2NW}} - V_{\mathrm{adj}}$ are
both positive semidefinite.
\end{proposition}

The weight matrix of $\widehat V_{\mathrm{C2NW}}$, unlike those of $\widehat
V_{\mathrm{CHS}}$ and $\widehat V^\star_{\mathrm{ET}}$, is not entrywise bounded: the
conservative term $2\sum_t y_ty_t'$ sits inside the bracket summed over $m$, so every lag
deposits on the same within-period pairs, which accumulate weight
$2\sum_{m=1}^M\omega(m,M)=M$. 
The lag terms are innocuous by contrast: $\omega(m,M)\sum_t y_ty_{t+m}'$ deposits on pairs at temporal distance \emph{exactly} $m$, so distinct lags occupy disjoint entries and each receives a single kernel weight --- which is why $W_n^{\mathrm{CHS}}$ has entries in $[0,2]$ irrespective of $M$.
Since the variance bound of \Cref{lem:quadratic_LLN} scales with the
square of the largest entry, $\widehat V_{\mathrm{C2NW}}$ requires an additional condition.

\begin{assumption}
\label{asmp:vcon_converge_c2nw}
$\frac{M^2 n}{\lambda_n^2} \sum_{s \geq 0} c_n(s, 0; 2)\, \theta_{n,s}^{1 - 4/p} \to 0$.
\end{assumption}

The $O(M)$ entries are confined to pairs at distance $0$, so
\Cref{asmp:vcon_converge_c2nw} evaluates $c_n(\cdot)$ at bandwidth $0$ rather than $M$;
since $c_n(s,0;2)=0$ for $s\ge1$, the sum collapses to its $s=0$ term. In the balanced
panel $c_n(0,0;2)\asymp(G+T)^2$ and $\lambda_n\asymp T^2G$, so the left-hand side is
$\asymp M^2/T^2$ --- the same order as that of \Cref{asmp:vcon_converge}, and likewise
reducing to $M=o(T)$. 
The two estimators $\widehat{V}_{\mathrm{C2NW}}$ and $\widehat{V}^\star_{\mathrm{ET}}$ therefore require the same bandwidth condition in
the balanced panel, and the case for $\widehat V^\star_{\mathrm{ET}}$ rests on the
efficiency ordering of \Cref{rem:efficiency-ordering}, not on weaker regularity.

\begin{lemma}[Consistency of the band-limited estimators]
\label{lem:both-consistent}
Suppose \Cref{asmp:psi_dependence} holds and
$\lambda_{\min}(V_\bullet)/\lambda_n \geq 1 + o(1)$. Then
\[
V_\bullet^{-1} \widehat V_\bullet \xrightarrow{p} I_v ,
\]
for $\bullet=\mathrm{band}$ under \Cref{asmp:vcon_converge}, and for
$\bullet=\mathrm{C2NW}$ under Assumptions~\ref{asmp:vcon_converge}
and~\ref{asmp:vcon_converge_c2nw}.
\end{lemma}

\begin{proposition}[Consistency of the exact eigenvalue-truncation estimator]
\label{prop:exact-opt-consistent}
Suppose Assumptions~\ref{asmp:psi_dependence} and~\ref{asmp:vcon_converge} hold,
$\lambda_{\min}(V_{\mathrm{ET}}^\star)/\lambda_n\ge 1+o(1)$, and $Mn/\lambda_n\to0$
(equivalently $M/T\to0$ in the balanced panel, where $\lambda_n\asymp T^2G$). Then
\[
(V_{\mathrm{ET}}^\star)^{-1}\,\widehat V_{\mathrm{ET}}^\star \xrightarrow{p} I_v.
\]
\end{proposition}

\begin{proposition}
\label{prop:vadj_to_vtrue}
Under Assumptions~\ref{asmp:psi_dependence} and~\ref{asmp:vadj_to_vtrue},
$V_{\mathrm{true}}^{-1} V_{\mathrm{adj}} \to I_v$.
\end{proposition}

\begin{remark}[Efficiency ordering]
\label{rem:efficiency-ordering}
$\widehat V_{\mathrm{C2NW}}$ overshoots by placing weight on the group-symmetric subspace, where $W_n^{\mathrm{CHS}}$ needs no correction, rather than on the group-demeaned subspace where its negative eigenvalues live. 
Since $2\sum_{m=1}^M\omega(m,M)=M$ for the triangular kernel, the within-period pairs of $\widehat V_{\mathrm{C2NW}}$ carry weight $M$, so
\[
A_n^{\mathrm{C2NW}} = J_T\otimes I_G + (C_M+M\,I_T)\otimes J_G
= \big[J_T+G\,C_M+GM\,I_T\big]\otimes P_G \;+\; J_T\otimes Q_G ,
\]
while $A_n^{\mathrm{band}} = \big[J_T+(G-1)C_M\big]\otimes P_G + J_T\otimes Q_G$ from the
proof of \Cref{cor:chs_optimal}. The two coincide on the $Q_G$ block, and
\[
A_n^{\mathrm{C2NW}}-A_n^{\mathrm{band}} = \big(C_M+GM\,I_T\big)\otimes P_G \;\succeq\; 0,
\]
positive definite on the group-symmetric block, so $V_{\mathrm{C2NW}}-V_{\mathrm{band}}
\succeq0$. Since $A_n^{\mathrm{band}}-A_n^\star=(C_M-B^-)\otimes Q_G$ has operator norm at
most $M+1$, we have $V_{\mathrm{band}}-V^\star_{\mathrm{ET}} = O(Mn) = o(\lambda_n)$, and
$V_{\mathrm{C2NW}}\succeq V^\star_{\mathrm{ET}}$ up to that order. 
\end{remark}

Under the conditions of Proposition~\ref{prop:vadj_to_vtrue}, the auxiliary condition
$\lambda_{\min}(V_\bullet)/\lambda_n \geq 1 + o(1)$ holds for
$\bullet\in\{\mathrm{C2NW},\mathrm{ET}^\star\}$, since
$\lambda_{\min}(V_\bullet)/\lambda_{\min}(V_{\mathrm{adj}}) \geq 1$ by
Proposition~\ref{prop:both-conservative} and
$\lambda_{\min}(V_{\mathrm{adj}})/\lambda_n = 1 + o(1)$ by
Proposition~\ref{prop:vadj_to_vtrue}. In the running AR(1) examples with
$\rho\in(0,1)$ and the triangular kernel, all three parts of
Assumption~\ref{asmp:vadj_to_vtrue} are implied by $\rho^m$ decay and the balanced
panel structure.


\section{Proofs} \label{sec:proofs}

Let $N_{n}\left(i;s\right)$ denote the set of observations that are within distance $s$ of observation $i$, and let $N_{n}^{\partial}\left(i;s\right)$ denote the set of observations that are exactly at distance $s$ from observation $i$:
\begin{align*}
N_{n}\left(i;s\right) & :=\left\{ j\in N_{n}:d_{n}\left(i,j\right)\leq s\right\}; \text{ and } \\
N_{n}^{\partial}\left(i;s\right) & :=\left\{ j\in N_{n}:d_{n}\left(i,j\right)=s\right\}.  \tag{\stepcounter{equation}\theequation}
\end{align*}

For matrices $A,B$, the notation $A \geq B$ means $A-B$ is positive semidefinite (psd).

\subsection{Proofs of Theorems}

\begin{proof}[Proof of \Cref{thm:impossibility}]
Fix any $A_n \in \mathrm{Sym}_n$ and any
$\Sigma_0 \in \mathcal{S}_n$, $v_n \in \mathcal{M}_n$ satisfying the hypotheses. Define
$\Sigma_1 := \Sigma_0 + v_n v_n' \in \mathcal{S}_n$ and consider
\[
P_n: Y_n \sim N(0, \Sigma_1), \qquad Q_n: Y_n \sim N(v_n, \Sigma_0).
\]
Both have finite second moments, mean in $\mathcal{M}_n$ (since
$0, v_n \in \mathcal{M}_n$), and covariance in $\mathcal{S}_n$, so both lie in
$\mathcal{P}_n(\mathcal{M}_n, \mathcal{S}_n)$. The expectations of the quadratic
estimator under each distribution are
\begin{align*}
\mathbb{E}_{P_n}[\widehat{V}_{A,n}] &= \mathrm{tr}(A_n \Sigma_1) + 0 = \mathrm{tr}(A_n \Sigma_0) + v_n' A_n v_n, \\
\mathbb{E}_{Q_n}[\widehat{V}_{A,n}] &= \mathrm{tr}(A_n \Sigma_0) + v_n' A_n v_n,
\end{align*}
which are identical: $\mathbb{E}_{P_n}[\widehat{V}_{A,n}] = \mathbb{E}_{Q_n}[\widehat{V}_{A,n}]$
for every symmetric $A_n$. However, the target values differ:
\[
\theta_n(P_n) - \theta_n(Q_n) = \mathrm{tr}(W_n(\Sigma_1 - \Sigma_0)) = v_n' W_n v_n.
\]
Let $\Delta := v_n' W_n v_n$. For any choice of $A_n$, since the estimator
expectation is common across $P_n$ and $Q_n$ while the target differs by
$\Delta$, the absolute bias must satisfy
\[
\big|\mathbb{E}_{P_n}[\widehat{V}_{A,n}] - \theta_n(P_n)\big| + \big|\mathbb{E}_{Q_n}[\widehat{V}_{A,n}] - \theta_n(Q_n)\big| \geq |\Delta|,
\]
so $\sup_{P \in \{P_n, Q_n\}} |\mathbb{E}_P[\widehat{V}_{A,n}] - \theta_n(P)| \geq |\Delta|/2$,
and hence the supremum over the full model is at least $|\Delta|/2$. 
\end{proof}

\begin{proof}[Proof of \Cref{thm:conservative}]
By \eqref{eq:plug-decomp}, conservativeness is
\[
\tr\!\bigl((A_n-W_n)\Sigma\bigr)+\mu'A_n\mu\geq0
\quad\text{for all }(\mu,\Sigma)\in\mathcal{M}_n\times\mathcal{S}_n.
\]
Since $0\in\mathcal{M}_n$ as $\mathcal{M}_n$ is a linear subspace of $\mathbb{R}^n$ and $0\in\mathcal{S}_n$ as $\mathcal{S}_n$ is a closed convex cone, taking $\mu=0$ and
$\Sigma=0$ separately decouples the joint requirement into
\eqref{eq:C1}--\eqref{eq:C2}. Sufficiency is immediate.
\end{proof}

\begin{proof}[Proof of \Cref{thm:eig-truncation}]
Throughout, $\mathcal M_n=\R^n$, so (C2) reads $A_n\succeq0$ and
$\mathcal R_n(A_n)=\sup_{\|\mu\|=1}\mu'(A_n-W_n)\mu=\lambda_{\max}(A_n-W_n)$.

\emph{(i).} Let $A_n$ be uniformly conservative and set $\Delta:=A_n-W_n$. Two lower
bounds combine. First, take a unit eigenvector $u$ with $W_nu=\lambda_n^{\min}u$; by
(C2), $u'A_nu\ge0$, so
\[
\mathcal R_n(A_n)\ \ge\ u'\Delta u\ =\ u'A_nu-\lambda_n^{\min}\ \ge\ -\lambda_n^{\min}.
\]
Second, since $\mathcal S_n\neq\{0\}$ pick $\Sigma_0\in\mathcal S_n\setminus\{0\}$; it is
positive semidefinite and nonzero, so if $\mathcal R_n(A_n)<0$, i.e.\ $\Delta\prec0$,
then $\tr(\Delta\Sigma_0)<0$, contradicting $\Delta\in\mathcal S_n^*$. Hence
$\mathcal R_n(A_n)\ge0$. Combining, $\mathcal R_n(A_n)\ge(-\lambda_n^{\min})_+$. The
truncation attains this bound and is feasible: $A_n^\star-W_n=W_n^-\succeq0\in
\mathcal S_n^*$ gives (C1), $A_n^\star=W_n^+\succeq0$ gives (C2), and
$\mathcal R_n(A_n^\star)=\lambda_{\max}(W_n^-)=(-\lambda_n^{\min})_+$. So $A_n^\star$
solves~\eqref{eq:opt-min-corr}.

\medskip
In (ii) and (iii), \eqref{eq:containment} enters only through the memberships
$q_{n,k}q_{n,k}'\in\mathcal S_n$, which turn (C1) into the scalar inequalities
$q_{n,k}'(A_n-W_n)q_{n,k}\ge0$.

\emph{The feasible box.} Under~\eqref{eq:containment}, the uniformly conservative
members of $\mathcal A_n^{\mathrm{sp}}$ are exactly
\begin{equation}
\{A_n(a):a_j\ge\lambda_{n,j}^{-}\ \forall j\},
\qquad\text{whose Loewner-least element is }A_n(\lambda^{-})=W_n^{+}=A_n^\star .
\label{eq:box}
\end{equation}
Indeed, if $a_j\ge\lambda_{n,j}^{-}$ for all $j$ then $a\ge0$, so
$\Delta(a)\succeq0\in\Sym_n^+\subseteq\mathcal S_n^*$, and
$A_n(a)=\sum_j(\lambda_{n,j}+a_j)q_{n,j}q_{n,j}'\succeq0$: both (C1) and (C2) hold.
Conversely let $A_n(a)$ be uniformly conservative. If $\lambda_{n,k}\le0$, (C2) gives
$a_k\ge-\lambda_{n,k}=\lambda_{n,k}^{-}$. If $\lambda_{n,k}>0$, then (C1) at
$q_{n,k}q_{n,k}'\in\mathcal S_n$ gives
$a_k=\tr\bigl(\Delta(a)\,q_{n,k}q_{n,k}'\bigr)\ge0=\lambda_{n,k}^{-}$.

\emph{(ii)(a).} $\tr(A_n(a)\Psi)=\tr(W_n\Psi)+\sum_ja_j\psi_j$ with $\psi_j\ge0$ is
coordinatewise nondecreasing on the box~\eqref{eq:box}, hence minimized at its corner
$a=\lambda^{-}$; linearity in $a$ makes the minimizer unique exactly when every
$\psi_j>0$.

\emph{(ii)(b).} Let $A_n$ be uniformly conservative and fix $k$. If $\lambda_{n,k}>0$,
(C1) at $q_{n,k}q_{n,k}'\in\mathcal S_n$ gives
$q_{n,k}'A_nq_{n,k}\ge q_{n,k}'W_nq_{n,k}=\lambda_{n,k}=\lambda_{n,k}^{+}$; if
$\lambda_{n,k}\le0$, (C2) gives $q_{n,k}'A_nq_{n,k}\ge0=\lambda_{n,k}^{+}$. Hence, for
$\Psi$ diagonal in the eigenbasis,
$\tr(A_n\Psi)=\sum_j\psi_j\,q_{n,j}'A_nq_{n,j}\ge\sum_j\psi_j\lambda_{n,j}^{+}
=\tr(A_n^\star\Psi)$.

\emph{(ii)(c).} 
Recall $A_n(a)=W_n+\Delta(a)$ from \Cref{def:spectral-aligned}, so at
$\Psi=q_{n,k}q_{n,k}'$ every spectral candidate has level
$\tr(A_n(a)\Psi)=\lambda_{n,k}+a_k$, while $\tr(A_n^\star\Psi)=\lambda_{n,k}$. The
claim is therefore equivalent to exhibiting a uniformly conservative $A_n(a)$ with
$a_k<0$. By the box~\eqref{eq:box}, no such $a$ exists when~\eqref{eq:containment}
holds; the construction below extracts an
$a''$ with $a_k''<0$ when $q_{n,k}q_{n,k}'\notin\mathcal S_n$.

Let $\mathcal C:=\{d(\Sigma):\Sigma\in\mathcal S_n\}\subseteq\R^n_+$,
where $d(\Sigma):=\bigl(q_{n,1}'\Sigma q_{n,1},\dots,q_{n,n}'\Sigma q_{n,n}\bigr)'$; it
is a convex cone, being a linear image of one, and $\overline{\mathcal C}$ is a closed
convex cone.
Here $\{e_j\}_{j=1}^n$ is the standard basis of $\R^n$ in the eigen-index $j$, so
that $e_k=d(q_{n,k}q_{n,k}')$ is the diagonal image of the $k$-th eigenprojector.

We first claim $e_k\notin\overline{\mathcal C}$. If it were, there would be
$\Sigma^{(m)}\in\mathcal S_n$ with $d(\Sigma^{(m)})\to e_k$; then
$q_{n,k}'\Sigma^{(m)}q_{n,k}\to1$, so rescaling within the cone we may take
$q_{n,k}'\Sigma^{(m)}q_{n,k}=1$, whence
$\tr(\Sigma^{(m)})=\sum_jq_{n,j}'\Sigma^{(m)}q_{n,j}\to1$. The sequence is bounded, so
along a subsequence $\Sigma^{(m)}\to\Sigma^\star$, and $\Sigma^\star\in\mathcal S_n$
because $\mathcal S_n$ is closed, with
$\tr(\Sigma^\star)=q_{n,k}'\Sigma^\star q_{n,k}=1$. For $\Sigma\succeq0$ of unit trace,
$q'\Sigma q\le1$ with equality only at $\Sigma=qq'$; hence
$\Sigma^\star=q_{n,k}q_{n,k}'\in\mathcal S_n$, a contradiction.

Separation therefore supplies $a\in\R^n$ with $a'd\ge0$ for all
$d\in\overline{\mathcal C}$ and $a_k=a'e_k<0$. Write
$\mathcal C^*:=\{b:b'd\ge0\ \forall d\in\overline{\mathcal C}\}$, a convex cone
containing $\R^n_+$ because $\overline{\mathcal C}\subseteq\R^n_+$. Since
$\tr\bigl(\Delta(b)\Sigma\bigr)=b'd(\Sigma)$ for every $\Sigma$, we have
$\Delta(b)\in\mathcal S_n^*$ whenever $b\in\mathcal C^*$.

The separating $a$ delivers the negative $k$-th coordinate we need but may violate (C2) off coordinate $k$; we correct this without disturbing the sign of $a_k$ by adding a small nonnegative shift in the remaining coordinates.
Fix $0<\varepsilon<\lambda_{n,k}/|a_k|$, set $\delta:=\varepsilon\max_j|a_j|$, and let
\[
a'':=\varepsilon a+\lambda^{-}+\delta\textstyle\sum_{j\neq k}e_j\ \in\ \mathcal C^* ,
\]
the membership because $\mathcal C^*$ is a convex cone containing $\varepsilon a$ and
$\R^n_+$. Then $\Delta(a'')\in\mathcal S_n^*$, giving (C1). For (C2): at $j=k$,
$\lambda_{n,k}^{-}=0$ and $\lambda_{n,k}+a_k''=\lambda_{n,k}+\varepsilon a_k>0$; for
$j\neq k$, $\lambda_{n,j}+a_j''=\lambda_{n,j}^{+}+\varepsilon a_j+\delta\ge0$ since
$\delta\ge\varepsilon|a_j|$. So $A_n(a'')\in\mathcal A_n^{\mathrm{sp}}$ is uniformly
conservative, yet at $\Psi=q_{n,k}q_{n,k}'$,
\[
\tr\bigl(A_n(a'')\Psi\bigr)=\lambda_{n,k}+\varepsilon a_k<\lambda_{n,k}
=\tr(A_n^\star\Psi),
\]
which is the assertion of (c). 

\emph{(iii).} By~\eqref{eq:containment},
$\Sigma\in\mathcal S_n^{W\text{-}\mathrm{diag}}$ lies in $\mathcal S_n$, so any $P_n$
with moments $(\mu,\Sigma)$, $\mu\in\R^n$, lies in $\mathcal P_n(\R^n,\mathcal S_n)$:
\Cref{asmp:no-skew} applies and the MSE is given by~\eqref{eq:mse-clean}. Write
$\Sigma=\Delta(s)$, so $\Sigma q_{n,j}=s_jq_{n,j}$ with $s_j\ge0$, and set
$m_j:=q_{n,j}'\mu$ and $c_j:=\lambda_{n,j}+a_j$. On the box~\eqref{eq:box}, $a\ge0$ and
$c_j\ge\lambda_{n,j}^{+}\ge0$. Diagonality in $\{q_{n,j}\}$ collapses each term
of~\eqref{eq:mse-clean}:
\begin{align*}
\tr((A_n(a)-W_n)\Sigma)=\sum_ja_js_j,\quad
\mu'A_n(a)\mu=\sum_jc_jm_j^2, \\
\tr((A_n(a)\Sigma)^2)=\sum_jc_j^2s_j^2,\quad
\mu'A_n(a)\Sigma A_n(a)\mu=\sum_jc_j^2s_jm_j^2 ,
\end{align*}
so that, with $b(a):=\sum_ja_js_j+\sum_jc_jm_j^2$,
\[
\mathrm{MSE}_n(A_n(a);\mu,\Sigma)=b(a)^2+2\sum_jc_j^2s_j^2+4\sum_jc_j^2s_jm_j^2 .
\]
The quantity $b(a)$ is the bias, nonnegative by \Cref{thm:conservative} and, on the
box, termwise. For each $\ell$,
\[
\frac{\partial\,\mathrm{MSE}_n}{\partial a_\ell}
=2b(a)\,(s_\ell+m_\ell^2)+2c_\ell\bigl(2s_\ell^2+4s_\ell m_\ell^2\bigr)\ \ge\ 0 ,
\]
so $\mathrm{MSE}_n$ is coordinatewise nondecreasing on the box and is minimized at its
corner $a=\lambda^{-}$, i.e.\ at $A_n^\star$. If every $s_j>0$ then
$2s_\ell^2+4s_\ell m_\ell^2>0$ and $c_\ell>0$ off the corner, so the derivative is
strictly positive there and the minimizer is unique.
\end{proof}

\begin{proof}[Proof of \Cref{thm:clt}]
It suffices to verify the conditions of KMS Theorem 3.2, using notation in this setup. 
Our \Cref{asmp:psi_dependence} is identical to KMS Assumption 2.1. 
$Y_{n,i}$ satisfying that condition implies that $\nu^{\prime}Y_{n,i}$ also satisfies that condition.
\Cref{asmp:psi_dependence}(c)'s finite moment condition satisfies KMS Assumption 3.3. 
Observe that, since $\left\Vert \nu\right\Vert =1$, we have $\sigma_{n}^{2}\geq \lambda_{n}$ as the worst case $\sigma_{n}^{2}$ loads all weight on a combination of variance components that results in the smallest eigenvalue. 
Since $\sigma_{n}^{2}\geq \lambda_{n}$, for some arbitrary constant $C>0$, 
\begin{align*}
\frac{n}{\sigma_{n}^{2+k}}\sum_{s\geq0}c_{n}\left(s,m_{n};k\right)\theta_{n,s}^{1-\frac{2+k}{p}} & \leq\frac{Cn}{\lambda_{n}^{1+k/2}}\sum_{s\geq0}c_{n}\left(s,m_{n};k\right)\theta_{n,s}^{1-\frac{2+k}{p}}\xrightarrow{a.s.}0, \text{ and } \tag{\stepcounter{equation}\theequation}\\
\frac{n^{2}\theta_{n,m_{n}}^{1-(1/p)}}{\sigma_{n}} & \leq\frac{Cn^{2}\theta_{n,m_{n}}^{1-(1/p)}}{\lambda_{n}^{1/2}}\xrightarrow{a.s.}0.  \tag{\stepcounter{equation}\theequation} 
\end{align*}

It then remains to show that the setting in this paper satisfies the corresponding definitions in KMS. 
First, normalize $N_{n}\left(j;s\right)=\emptyset$ for $s<0$, and observe that, when using $s=0$, 
\begin{equation} \label{eq:s0_kms_match}
\frac{1}{n}\sum_{i\in N_{n}}\left|N_{n}^{\partial}\left(i;0\right)\right|^{k}\leq\frac{1}{n}\sum_{i\in N_{n}}\left(\left|N_{g(i)}^{G}\right|+\left|N_{t(i)}^{T}\right|\right)^{k}=\delta_{n}^{\partial}\left(0;k\right).
\end{equation}

There are two cases to consider: when $m=0$ and when $m\geq1$.
When $m=0$, observe that, for all $s \geq 0$, we obtain 
\begin{equation}
\frac{1}{n}\sum_{i\in N_{n}}\left|N_{n}\left(i;0\right)\right|^{k} \leq\frac{1}{n}\sum_{i\in N_{n}}\left(\left|N_{g(i)}^{G}\right|+\left|N_{t(i)}^{T}\right|\right)^{k}=\Delta_{n}\left(s,0;k\right),
\end{equation}
so the $\delta^\partial_n(0;k)$ and $\Delta_n(s,0;k)$ defined here place an upper bound on the corresponding definitions in KMS for $m=0$.

Now, consider the case of $m\geq 1$.
The case of $s=0$ is derived in \Cref{eq:s0_kms_match}.
When $s \geq 1$, the corresponding definitions are also satisfied as
\begin{equation}
\frac{1}{n}\sum_{i\in N_{n}}\left|N_{n}^{\partial}\left(i;s\right)\right|^{k} =\frac{1}{n}\sum_{i\in N_{n}}\left(\left|N_{t(i)+s}^{T}\right|+\left|N_{t(i)-s}^{T}\right|\right)^{k} =\delta_{n}^{\partial}\left(s;k\right), \text{ and }
\end{equation}
\begin{align*}
 & \frac{1}{n}\sum_{i\in N_{n}}\max_{j\in N_{n}^{\partial}\left(i;s\right)}\left|N_{n}\left(i;m\right)\backslash N_{n}\left(j;s-1\right)\right|^{k}\\
 & =\frac{1}{n}\sum_{i\in N_{n}}\max_{j\in N_{n}^{\partial}\left(i;s\right)}\left|\sum_{h=0}^{m}\left|N_{t(i)+h}^{T}\right|+\sum_{h=0}^{m}\left|N_{t(i)-h}^{T}\right|-\sum_{h=0}^{s-1}\left|N_{t(i)+h}^{T}\right|-\sum_{h=0}^{s-1}\left|N_{t(i)-h}^{T}\right|\right|^{k}\\
 & =\frac{1}{n}\sum_{i\in N_{n}}\left(\sum_{h=0}^{m}\left|N_{t(i)+h}^{T}\right|+\sum_{h=0}^{m}\left|N_{t(i)-h}^{T}\right|-\sum_{h=0}^{s-1}\left|N_{t(i)+h}^{T}\right|-\sum_{h=0}^{s-1}\left|N_{t(i)-h}^{T}\right|\right)^{k}\\
 & =\frac{1}{n}\sum_{i\in N_{n}}\left(\sum_{h=s}^{m}\left|N_{t(i)+h}^{T}\right|+\sum_{h=s}^{m}\left|N_{t(i)-h}^{T}\right|\right)^{k}=\Delta_{n}\left(s,m;k\right). \tag{\stepcounter{equation}\theequation}
\end{align*}

The expression $c_{n}\left(s,m;k\right)=\inf_{\alpha>1}\left[\Delta_{n}\left(s,m;k\alpha\right)\right]^{1/\alpha}\left[\delta_{n}^{\partial}\left(s;\frac{\alpha}{\alpha-1}\right)\right]^{1-\frac{1}{\alpha}}$ is identical to KMS. 
Consequently, the conditions of KMS Theorem 3.2 are satisfied under the Assumptions \ref{asmp:psi_dependence} and \ref{asmp:clt}. 
\end{proof}

\begin{proof}[Proof of \Cref{thm:panel_dep_consistent_main}]
(a) By Corollary~\ref{cor:chs_optimal},
$(W_n^{\mathrm{CHS}})^+ - W_n^{\mathrm{CHS}} = B^-\otimes Q_G\succeq0$, and
$(W_n^{\mathrm{CHS}})^+\succeq0$ on $\mathcal M_n=\mathbb R^n$. Hence the bias decomposition
\[
V_{\mathrm{ET}}^\star - V_{\mathrm{adj}}
= \operatorname{tr}\!\big((B^-\otimes Q_G)\,\Sigma_n\big)
+ \mu_n'(W_n^{\mathrm{CHS}})^+\mu_n
\]
has both terms nonnegative: the first is a trace of positive semidefinite matrices,
the second is nonnegative by positive semidefiniteness of $(W_n^{\mathrm{CHS}})^+$.

(b) This is Proposition~\ref{prop:exact-opt-consistent} in
Appendix~\ref{sec:consistency_appendix}. The argument writes
$\widehat V_{\mathrm{ET}}^\star = \widehat V_{\mathrm{band}} - Y_n'D Y_n$, where
$\widehat V_{\mathrm{band}} = Y_n'(W_n^{\mathrm{CHS}} + C_M\otimes Q_G)Y_n$ is a
band-limited proxy consistent by the quadratic LLN of
Lemma~\ref{lem:quadratic_LLN}, and $D = (C_M - B^-)\otimes Q_G$ has
$\|D\|_{\mathrm{op}}\le M+1$, so $Y_n'DY_n = O_p(Mn) = o_p(\lambda_n)$.

(c) This is Proposition~\ref{prop:vadj_to_vtrue} in
Appendix~\ref{sec:consistency_appendix}: under Assumption~\ref{asmp:vadj_to_vtrue} the
discarded long-lag and boundary terms vanish relative to $\lambda_n$.
\end{proof}

\subsection{Proofs of Propositions}

\begin{proof}[Proof of \Cref{prop:estimand-vs-estimate}]
(i)$\Leftrightarrow$(iii) is the definition of positive
semidefiniteness. For (ii)$\Leftrightarrow$(iii): the plug-in $A_n=W_n$ satisfies
(C1) trivially, since $A_n-W_n=0\in\mathcal S_n^*$ for every cone. By
\Cref{thm:conservative}, conservativeness on
$\mathcal P_n(\mathcal M_n,\mathcal S_n)$ is therefore equivalent to (C2) alone,
$\mu'W_n\mu\ge0$ for all $\mu\in\mathcal M_n$, independently of $\mathcal S_n$. At
$\mathcal M_n=\mathbb R^n$ this reads $W_n\succeq0$.
\end{proof}

\begin{proof}[Proof of \Cref{prop:union-cone}]
\emph{(i).} Let $\Sigma_i=\Sigma_i^{(1)}+\Sigma_i^{(2)}$ with $\Sigma_i^{(1)}\in\mathcal S_n$
and $\Sigma_i^{(2)}\in\mathcal S_n^{W\text{-}\mathrm{diag}}$, $i=1,2$, and let $t\in[0,1]$.
Then
\[
t\Sigma_1+(1-t)\Sigma_2
=\underbrace{\bigl[t\Sigma_1^{(1)}+(1-t)\Sigma_2^{(1)}\bigr]}_{\in\ \mathcal S_n}
+\underbrace{\bigl[t\Sigma_1^{(2)}+(1-t)\Sigma_2^{(2)}\bigr]}_{\in\ \mathcal S_n^{W\text{-}\mathrm{diag}}}
\ \in\ \mathcal S_n^{\oplus},
\]
each bracket lying in its own cone by the convexity of that cone; likewise
$c(\Sigma^{(1)}+\Sigma^{(2)})=c\Sigma^{(1)}+c\Sigma^{(2)}\in\mathcal S_n^{\oplus}$ for $c\ge0$.
Convexity of the summands is doing the work that the union of the two cones, being nonconvex
in general, cannot do. Also $\mathcal S_n^{\oplus}\subseteq\Sym_n^+$ because $\Sym_n^+$ is
closed under addition, and $0$ belongs to each summand, so $\mathcal S_n^{\oplus}$ contains
both $\mathcal S_n$ and $\mathcal S_n^{W\text{-}\mathrm{diag}}$: it is admissible. For
minimality, let $\mathcal T$ be admissible, and take
$\Sigma_1\in\mathcal S_n\subseteq\mathcal T$ and
$\Sigma_2\in\mathcal S_n^{W\text{-}\mathrm{diag}}\subseteq\mathcal T$. Convexity of
$\mathcal T$ gives $\tfrac12\Sigma_1+\tfrac12\Sigma_2\in\mathcal T$, and the cone property
gives $\Sigma_1+\Sigma_2\in\mathcal T$. Hence $\mathcal S_n^{\oplus}\subseteq\mathcal T$.

\emph{(ii).} If $D\in\mathcal S_n^*\cap(\mathcal S_n^{W\text{-}\mathrm{diag}})^*$ then
$\tr\bigl(D(\Sigma_1+\Sigma_2)\bigr)=\tr(D\Sigma_1)+\tr(D\Sigma_2)\ge0$; conversely, taking
$\Sigma_2=0$ and then $\Sigma_1=0$ in $D\in(\mathcal S_n^{\oplus})^*$ delivers the two
memberships. The stated form of $(\mathcal S_n^{W\text{-}\mathrm{diag}})^*$ follows from
$\tr\bigl(D\Delta(s)\bigr)=\sum_js_j\,q_{n,j}'Dq_{n,j}$, which is nonnegative for every
$s\ge0$ if and only if every coefficient is. The characterization of uniform conservativeness
follows because (C2) does not involve the cone. Finally, if $\mathcal T$ is admissible then
$\mathcal S_n^{\oplus}\subseteq\mathcal T$ by (i), and dualization reverses inclusion, so
$\mathcal T^*\subseteq(\mathcal S_n^{\oplus})^*$; since (C1) requires
$A_n-W_n\in\mathcal T^*$, the feasible set under $\mathcal T$ is contained in the feasible set
under $\mathcal S_n^{\oplus}$.
\end{proof}

\begin{proof}[Proof of \Cref{prop:conicity}]
In all three parts the feasible set is fixed: $A_n-W_n\in\mathcal S_n^*$ is a
conic constraint (the dual cone $\mathcal S_n^*$ is closed convex), and
$P_{\mathcal M_n}A_nP_{\mathcal M_n}\succeq0$ is a linear matrix inequality,
hence conic. It remains to treat the objectives.

\textit{(i).} Introduce $t\in\R$ and rewrite \eqref{eq:opt-min-corr} as $\min t$
subject to $\mu'(A_n-W_n)\mu\le t$ for all $\mu\in\mathcal{M}_n$ with
$\|\mu\|=1$, equivalently the linear matrix inequality
$t\,P_{\mathcal{M}_n}\succeq P_{\mathcal{M}_n}(A_n-W_n)P_{\mathcal{M}_n}$.
This is a linear objective subject to an linear matrix inequality (LMI) and the conic feasibility constraints, hence a convex conic program.

\textit{(ii).} For $\pi$-almost every $(\mu,\Sigma)$ the integrand
$\mathrm{MSE}_n(\,\cdot\,;\mu,\Sigma)$ is convex in $A_n$ by hypothesis. The
averaged objective
\[
\overline{\mathrm{MSE}}_n(A_n;\pi)
=\int \mathrm{MSE}_n(A_n;\mu,\Sigma)\,d\pi(\mu,\Sigma)
\]
is therefore an integral of convex functions, hence convex in $A_n$; the
finite-fourth-moment assumption guarantees the integral is finite for every
$A_n\in\mathrm{Sym}_n$. Minimizing a convex objective over the convex conic
feasible set is a convex conic program.

\textit{(iii).} By \eqref{eq:plug-decomp}, for each $(\mu,\Sigma)$ the level
$V_{A,n}(\mu,\Sigma)=\tr\!\big(A_n(\Sigma+\mu\mu')\big)$ is a linear functional of
$A_n$. Integrating against $\pi$ and using linearity of the trace,
\[
\overline V_n(A_n;\pi)
=\int \tr\!\big(A_n(\Sigma+\mu\mu')\big)\,d\pi(\mu,\Sigma)
=\tr\!\big(A_n\,\Psi_\pi\big),
\qquad \Psi_\pi:=\E_\pi[\Sigma+\mu\mu'],
\]
finite since $\pi$ has finite second moments, and linear in $A_n$. Minimizing a
linear objective over the conic feasible set is a linear conic program.
\end{proof}

\begin{proof}[Proof of \Cref{prop:both-conservative}]
Throughout write $a_i:=\E[Y_{n,i}\mid\mathcal C_n]$ and
$\bar y_t:=\E[y_t\mid\mathcal C_n]=\sum_{i\in N_t^T}a_i$.

\emph{$\widehat V_{\mathrm{ET}}^\star$.} By \Cref{cor:chs_optimal},
$(W_n^{\mathrm{CHS}})^+-W_n^{\mathrm{CHS}}=B^-\otimes Q_G\succeq0$ and
$(W_n^{\mathrm{CHS}})^+\succeq0$, so
\[
V_{\mathrm{ET}}^\star-V_{\mathrm{adj}}
=\operatorname{tr}\!\big((B^-\otimes Q_G)\,\Sigma_n\big)+\mu_n'(W_n^{\mathrm{CHS}})^+\mu_n\ \ge\ 0 ,
\]
the first term a trace of a product of positive semidefinite matrices, the second
nonnegative by $(W_n^{\mathrm{CHS}})^+\succeq0$.

\emph{$\widehat V_{\mathrm{C2NW}}$.} Replacing each empirical product in
$\widehat V_{\mathrm{C2NW}}$ by its conditional expectation and subtracting
$V_{\mathrm{adj}}$ of~\eqref{eq:vadj}, the covariance parts of the two cluster
double-sums and of the lag terms cancel against $V_{\mathrm{adj}}$, leaving the
within-cell covariance, the group/time \emph{mean} outer products, and the mean part
of the kernel-weighted block:
\begin{align*}
V_{\mathrm{C2NW}}-V_{\mathrm{adj}}
&= \underbrace{\sum_{i}\sum_{j\in N_{g(i)}^G} a_i a_j'}_{\text{(I) within-group means}}
 + \underbrace{\sum_{i}\sum_{j\in N_{t(i)}^T} a_i a_j'}_{\text{(II) within-time means}}
 + \underbrace{\sum_{i}\sum_{j\in N_{t(i),g(i)}^{T\cap G}}\Cov(Y_{n,i},Y_{n,j}\mid\mathcal C_n)}_{\text{(III) within-cell covariance}}\\
&\quad + \underbrace{\sum_{m=1}^M\omega(m,M)\sum_{t=1}^{T-m}
   \big(\bar y_t\bar y_{t+m}'+\bar y_{t+m}\bar y_t'\big)
   \;+\; 2\sum_{m=1}^M\omega(m,M)\sum_{t=1}^{T}\bar y_t\bar y_t'}_{\text{(IV) kernel-weighted mean block}} .
\end{align*}

Terms (I)--(III) are positive semidefinite: (I) telescopes to
$\sum_g\big(\sum_{i\in N_g^G}a_i\big)\big(\sum_{i\in N_g^G}a_i\big)'$ and (II)
likewise to $\sum_t\big(\sum_{i\in N_t^T}a_i\big)\big(\sum_{i\in N_t^T}a_i\big)'$,
each a sum of outer products; (III) equals
$\sum_{g,t}\Var\!\big(\sum_{i\in N_{t,g}^{T\cap G}}Y_{n,i}\bigm|\mathcal C_n\big)\succeq0$.
For (IV), fix $m$ and complete the square lag by lag. For each $t\le T-m$,
\[
\bar y_t\bar y_{t+m}'+\bar y_{t+m}\bar y_t'
=(\bar y_t+\bar y_{t+m})(\bar y_t+\bar y_{t+m})'
-\bar y_t\bar y_t'-\bar y_{t+m}\bar y_{t+m}' ,
\]
so the two cross terms are a square minus the two diagonal squares. The doubled
within-period term supplies those diagonal squares. For each $m$,
\[
2\sum_{t=1}^{T}\bar y_t\bar y_t'
-\sum_{t=1}^{T-m}\big(\bar y_t\bar y_t'+\bar y_{t+m}\bar y_{t+m}'\big)
=\sum_{t=1}^{T} c_{t,m}\,\bar y_t\bar y_t',
\qquad c_{t,m}\in\{0,1,2\},
\]
a nonnegative-weighted sum of outer products, since each period is used as an endpoint at most twice. Therefore
\[
\text{(IV)}
=\sum_{m=1}^M\omega(m,M)\!\left[
\sum_{t=1}^{T-m}(\bar y_t+\bar y_{t+m})(\bar y_t+\bar y_{t+m})'
+\sum_{t=1}^{T} c_{t,m}\,\bar y_t\bar y_t'\right]\ \succeq\ 0 ,
\]
every summand a positive semidefinite outer product and $\omega(m,M)\ge0$ for the
triangular kernel. Adding (I)--(IV) gives $V_{\mathrm{C2NW}}-V_{\mathrm{adj}}\succeq0$.
\end{proof}

\begin{proof}[Proof of \Cref{prop:exact-opt-consistent}]
By the Cramér--Wold device and the minimum-eigenvalue normalization, as in the proof
of Lemma~\ref{lem:both-consistent}, it suffices to show
$\frac{1}{\lambda_n}\nu'(\widehat V_{\mathrm{ET}}^\star-V_{\mathrm{ET}}^\star)\nu
\xrightarrow{p}0$ for every unit vector $\nu$; set $\tilde Y_{n,i}:=\nu'Y_{n,i}$, a
scalar $\psi$-dependent array with the dependence coefficients of $Y_n$.

\emph{Decomposition.} The proxy and the exact estimator differ by
$D:=(C_M-B^-)\otimes Q_G$, since
$\widehat V_{\mathrm{band}} - \widehat V_{\mathrm{ET}}^\star
= Y_n'\big((C_M-B^-)\otimes Q_G\big)Y_n$. Writing
$\widehat\Delta := \tilde Y_n'D\tilde Y_n$ and
$\Delta := \E[\widehat\Delta\mid\mathcal C_n]$,
\[
\nu'(\widehat V_{\mathrm{ET}}^\star-V_{\mathrm{ET}}^\star)\nu
= \nu'(\widehat V_{\mathrm{band}}-V_{\mathrm{band}})\nu - (\widehat\Delta-\Delta).
\]
By Lemma~\ref{lem:both-consistent} the first term is $o_p(\lambda_n)$, so it remains to
bound $\widehat\Delta-\Delta$.

\emph{Operator-norm bound on $D$.} Both $C_M$ and $B^- = (C_M-J_T)_+$ are positive
semidefinite. Since the spectral window of $C_M$ is the Fej\'er kernel,
$\lambda_{\max}(C_M)\le f_{\mathrm{Fej}}(0)=M+1$; and as $-J_T\preceq0$,
$\lambda_{\max}(B^-) = \big(\lambda_{\max}(C_M-J_T)\big)_+\le\lambda_{\max}(C_M)\le M+1$.
Hence $\lambda_{\max}(C_M-B^-)\le\lambda_{\max}(C_M)\le M+1$ and
$\lambda_{\min}(C_M-B^-)\ge-\lambda_{\max}(B^-)\ge-(M+1)$, so
$\|C_M-B^-\|_{\mathrm{op}}\le M+1$. As $\|Q_G\|_{\mathrm{op}}=1$,
\[
\|D\|_{\mathrm{op}} = \|C_M-B^-\|_{\mathrm{op}}\,\|Q_G\|_{\mathrm{op}}\le M+1.
\]

\emph{Conclusion.} By the variational characterization of the operator norm,
\[
|\widehat\Delta| = |\tilde Y_n'D\tilde Y_n|
\le \|D\|_{\mathrm{op}}\,\|\tilde Y_n\|^2
\le (M+1)\sum_{i\in N_n}\tilde Y_{n,i}^2 .
\]
By Assumption~\ref{asmp:psi_dependence}(c),
$\E[\tilde Y_{n,i}^2\mid\mathcal C_n]\le\bar c$ a.s.\ uniformly in $i,n$, so
$\E[|\widehat\Delta|\mid\mathcal C_n]\le(M+1)\bar c\,n = O(Mn)$ a.s. This gives both
$|\Delta|\le\E[|\widehat\Delta|\mid\mathcal C_n]=O(Mn)$ and, by the conditional Markov
inequality, $\widehat\Delta = O_p(Mn)$. Therefore
\[
\frac{|\widehat\Delta-\Delta|}{\lambda_n}
\le \frac{|\widehat\Delta|+|\Delta|}{\lambda_n}
= O_p\!\Big(\frac{Mn}{\lambda_n}\Big) = o_p(1),
\]
using $Mn/\lambda_n\to0$. Combining the two displays gives
$\frac{1}{\lambda_n}\nu'(\widehat V_{\mathrm{ET}}^\star-V_{\mathrm{ET}}^\star)\nu
\xrightarrow{p}0$ for each fixed unit $\nu$, the unconditional convergence following by
dominated convergence as in Lemma~\ref{lem:quadratic_LLN}. Exactly as in the proof of
Lemma~\ref{lem:both-consistent}, in fixed dimension $v$ this per-direction bound yields
$\lambda_n^{-1}\|\widehat V_{\mathrm{ET}}^\star-V_{\mathrm{ET}}^\star\|_{\mathrm{op}}\xrightarrow{p}0$,
whence $(V_{\mathrm{ET}}^\star)^{-1}\widehat V_{\mathrm{ET}}^\star\xrightarrow{p}I_v$ using
$\|(V_{\mathrm{ET}}^\star)^{-1}\|_{\mathrm{op}}\le\lambda_n^{-1}(1+o(1))$ from the hypothesis
$\lambda_{\min}(V_{\mathrm{ET}}^\star)/\lambda_n\ge1+o(1)$.
\end{proof}

\begin{proof}[Proof of \Cref{prop:vadj_to_vtrue}]
It suffices to show $\frac{1}{\lambda_n}\nu'(V_{\mathrm{true}}-V_{\mathrm{adj}})\nu = o(1)$
for any unit vector $\nu$; setting $\tilde Y_{n,i}=\nu'Y_{n,i}$, this is
Lemma~\ref{lemma:vadj_to_vtrue_scalar}, whose conditions are immediate from
Assumption~\ref{asmp:vadj_to_vtrue}.
\end{proof}

\subsection{Statements and Proofs of Lemmas}

\begin{proof}[Proof of \Cref{lem:chs-matrix-form}]
Read $\widehat V_{\mathrm{CHS}}=\sum_{i,j}(W_n^{\mathrm{CHS}})_{ij}Y_{n,i}Y_{n,j}'$, so that
each line of \eqref{eq:Vchs} is the quadratic form of a weight matrix on the
index pair $\big((t,g),(s,h)\big)$. With one observation per cell,
$N_{g(i)}^G=\{(g,s)\}_s$, $N_{t(i)}^T=\{(h,t)\}_h$,
$N_{t(i),g(i)}^{T\cap G}=\{(g,t)\}$, and $y_t=\sum_h Y_{(h,t)}$, so the five
lines map to
\begin{align*}
\text{within-group}\quad &\textstyle\sum_i\sum_{j\in N_{g(i)}^G}
  && \longmapsto\quad J_T\otimes I_G,\\
\text{within-time}\quad &\textstyle\sum_i\sum_{j\in N_{t(i)}^T}
  && \longmapsto\quad I_T\otimes J_G,\\
\text{within-cell}\quad &-\textstyle\sum_i\sum_{j\in N_{t(i),g(i)}^{T\cap G}}
  && \longmapsto\quad -\,I_T\otimes I_G,\\
\text{aggregated lag}\quad &\textstyle\sum_{m\ge1}\omega(m,M)\sum_t
  \big(y_ty_{t+m}'+y_{t+m}y_t'\big)
  && \longmapsto\quad (C_M-I_T)\otimes J_G,\\
\text{double-count}\quad &-\textstyle\sum_{m\ge1}\omega(m,M)\sum_t\sum_g
  \big(Y_{(g,t)}Y_{(g,t+m)}'+\text{tr.}\big)
  && \longmapsto\quad -\,(C_M-I_T)\otimes I_G,
\end{align*}
using $(C_M)_{ts}=\omega(|t-s|,M)$ and $\omega(0,M)=1$. Summing the five
contributions and collecting the $\otimes J_G$ and $\otimes I_G$ blocks,
\[
\underbrace{\big[I_T+(C_M-I_T)\big]}_{=\,C_M}\otimes J_G
\;+\;
\underbrace{\big[J_T-I_T-(C_M-I_T)\big]}_{=\,J_T-C_M}\otimes I_G
\;=\;
C_M\otimes J_G+(J_T-C_M)\otimes I_G,
\]
which rearranges to $J_T\otimes I_G+C_M\otimes(J_G-I_G)=W_n^{\mathrm{CHS}}$.
\end{proof}

\begin{proof} [Proof of \Cref{lem:quadratic_LLN}]
We adapt the proof of Proposition~4.1 of \citet{kojevnikov2021limit}, translating from 
the HAC weight notation $\omega_n(s)$ used there to the matrix entries $(A_n)_{ij}$ used here. 
Throughout, all probabilistic statements are conditional on $\mathcal{C}_n$; the unconditional statement follows by the dominated convergence theorem applied to the conditional  probabilities, since they are bounded by $1$ and converge to $0$ a.s.

\textbf{Step 1: Centering.} Define $z_{ij} := Y_{n,i} Y_{n,j} - \mathbb{E}[Y_{n,i} Y_{n,j} \mid \mathcal{C}_n]$, 
so that $\mathbb{E}[z_{ij} \mid \mathcal{C}_n] = 0$. Then
\[
\hat{V}_{A,n} - V_{A,n} = \sum_{i \in N_n} \sum_{j \in N_n} (A_n)_{ij}\, z_{ij}.
\]
By condition (iii), the inner sum is restricted to pairs $(i, j)$ with $d_n(i,j) \leq b_n$:
\[
\hat{V}_{A,n} - V_{A,n} = \sum_{i \in N_n} \sum_{j \in N_n : d_n(i,j) \leq b_n} (A_n)_{ij}\, z_{ij}.
\]

\textbf{Step 2: Conditional variance bound.} We compute
\[
\mathrm{Var}\!\left(\hat{V}_{A,n} - V_{A,n} \mid \mathcal{C}_n\right) = \sum_{(i,j,k,l)} (A_n)_{ij} (A_n)_{kl} \, \mathbb{E}[z_{ij} z_{kl} \mid \mathcal{C}_n],
\]
where the sum runs over quadruples $(i, j, k, l) \in N_n^4$ with $d_n(i,j) \leq b_n$ and 
$d_n(k, l) \leq b_n$. By condition (ii), $|(A_n)_{ij}(A_n)_{kl}| \leq \bar{a}_n^2$, so
\[
\mathrm{Var}\!\left(\hat{V}_{A,n} - V_{A,n} \mid \mathcal{C}_n\right) \leq \bar{a}_n^2 \sum_{(i,j,k,l)} \bigl|\mathbb{E}[z_{ij} z_{kl} \mid \mathcal{C}_n]\bigr|.
\]

Group the quadruples by the distance between the pairs $\{i, j\}$ and $\{k, l\}$:
\[
\mathcal{H}_n(s, b_n) := \bigl\{(i, j, k, l) \in N_n^4 : d_n(i,j) \leq b_n,\ d_n(k,l) \leq b_n,\ d_n(\{i,j\}, \{k,l\}) = s\bigr\},
\]
so that
\[
\mathrm{Var}\!\left(\hat{V}_{A,n} - V_{A,n} \mid \mathcal{C}_n\right) \leq \bar{a}_n^2 \sum_{s \geq 0} \sum_{(i,j,k,l) \in \mathcal{H}_n(s, b_n)} \bigl|\mathbb{E}[z_{ij} z_{kl} \mid \mathcal{C}_n]\bigr|.
\]

\textbf{Step 3: Bounding individual covariances.} The variables $z_{ij}$ do not involve
$A_n$, so the bounds here are free of $\bar a_n$. For $s \geq 1$, the pairs $\{i,j\}$ and
$\{k,l\}$ are at distance $s$, so by KMS Theorem A.1 (combined with
Assumption~\ref{asmp:psi_dependence}(a) and condition (i) above with $p > 4$),
\[
\bigl|\mathbb{E}[z_{ij} z_{kl} \mid \mathcal{C}_n]\bigr| \leq C_1 \cdot \theta_{n,s}^{1 - 4/p} \quad \text{a.s.},
\]
for a constant $C_1 > 0$ depending only on $p$ and the moment bound in (i). For 
$s = 0$, the bound is uniform: 
$|\mathbb{E}[z_{ij} z_{kl} \mid \mathcal{C}_n]| \leq C_1$ by Cauchy--Schwarz applied to the conditional 
fourth moment, which is finite by condition (i).

\textbf{Step 4: Counting quadruples.} By inequality (A.13) of KMS,
\[
|\mathcal{H}_n(s, b_n)| \leq 4 n\, c_n(s, b_n; 2),
\]
where $c_n(s, b_n; 2)$ is defined in equation (11) of the paper.

\textbf{Step 5: Combining.} Substituting Steps 3 and 4 into the variance bound from Step 2:
\[
\mathrm{Var}\!\left(\hat{V}_{A,n} - V_{A,n} \mid \mathcal{C}_n\right)
\leq 4 C_1 \, \bar{a}_n^2 \cdot n \sum_{s \geq 0} c_n(s, b_n; 2)\, \theta_{n,s}^{1 - 4/p}
\xrightarrow{\text{a.s.}} 0
\]
by condition (iv).

\textbf{Step 6: Conclusion.} By Chebyshev's inequality, for any $\varepsilon > 0$,
\[
\mathbb{P}\!\left(\bigl|\hat{V}_{A,n} - V_{A,n}\bigr| > \varepsilon \,\Big|\, \mathcal{C}_n\right)
\leq \frac{1}{\varepsilon^2}\,\mathrm{Var}\!\left(\hat{V}_{A,n} - V_{A,n} \mid \mathcal{C}_n\right)
\xrightarrow{\text{a.s.}} 0.
\]
Taking unconditional expectations and applying dominated convergence (the conditional probabilities are bounded by $1$), $\mathbb{P}\!\left(\bigl|\hat{V}_{A,n} - V_{A,n}\bigr| > \varepsilon\right) \to 0$. 
\end{proof}

\begin{proof}[Proof of \Cref{lem:both-consistent}]
By Slutsky's lemma and the Cramér--Wold device it suffices to show
$\frac{1}{\lambda_n}\nu'(\widehat V_\bullet - V_\bullet)\nu = o_p(1)$ for any unit vector
$\nu$; set $\tilde Y_{n,i}:=\nu'Y_{n,i}$, a scalar $\psi$-dependent array with the
dependence coefficients of $Y_n$. Using $2\sum_{m=1}^M\omega(m,M)=M$, decompose
\[
A_n^{\mathrm{C2NW}} = \underbrace{J_T\otimes I_G + C_M\otimes J_G}_{=:A_n^{\mathrm{lag}}}
\;+\; \underbrace{M\,I_T\otimes J_G}_{=:A_n^{0}} ,
\]
and apply \Cref{lem:quadratic_LLN} to each of
$A_n^{\mathrm{band}}/\lambda_n$, $A_n^{\mathrm{lag}}/\lambda_n$, and
$A_n^{0}/\lambda_n$; the conclusion for $\bullet=\mathrm{C2NW}$ follows by summing the two
latter limits. We verify the conditions for $\bullet \in \{\mathrm{C2NW}, \mathrm{band}\}$.
\begin{enumerate}[label=(\roman*)]
\item Moment bound: from \Cref{asmp:psi_dependence}(c).
\item Weight bound: the entries of $A_n^{\mathrm{band}} = W_n^{\mathrm{CHS}}+C_M\otimes Q_G$
and of $A_n^{\mathrm{lag}}$ are bounded by a constant, so $\bar a_n\asymp\lambda_n^{-1}$;
the entries of $A_n^{0}$ equal $M$, so $\bar a_n\asymp M/\lambda_n$.
\item Support bandwidth: $b_n = M$ for $A_n^{\mathrm{band}}$ and $A_n^{\mathrm{lag}}$,
whose nonzero entries lie on pairs $(i,j)$ with $g(i)=g(j)$, $t(i)=t(j)$, or
$|t(i)-t(j)|\le M$; and $b_n=0$ for $A_n^{0}$, whose nonzero entries lie only on pairs with
$t(i)=t(j)$.
\item Variance summability: for $A_n^{\mathrm{band}}$ and $A_n^{\mathrm{lag}}$, condition
(iv) of \Cref{lem:quadratic_LLN} reads, up to a constant,
$\frac{n}{\lambda_n^2}\sum_{s\ge0}c_n(s,M;2)\theta_{n,s}^{1-4/p}\to0$, which is
\Cref{asmp:vcon_converge}. For $A_n^{0}$ it reads
$\frac{M^2n}{\lambda_n^2}\sum_{s\ge0}c_n(s,0;2)\theta_{n,s}^{1-4/p}\to0$, which is
\Cref{asmp:vcon_converge_c2nw}.
\end{enumerate}
\Cref{lem:quadratic_LLN} then gives, for each fixed unit $\nu$,
$\frac{1}{\lambda_n}\,\nu'(\widehat V_\bullet - V_\bullet)\nu\xrightarrow{p}0$. Because $v$
is fixed, evaluating this at the coordinate directions $e_a$ and the pairwise sums
$e_a+e_b$ controls every entry of the symmetric $v\times v$ matrix
$\widehat V_\bullet-V_\bullet$, so
$\lambda_n^{-1}\|\widehat V_\bullet-V_\bullet\|_{\mathrm{op}}\xrightarrow{p}0$. Writing
$V_\bullet^{-1}\widehat V_\bullet-I_v=V_\bullet^{-1}(\widehat V_\bullet-V_\bullet)$ and using
$\|V_\bullet^{-1}\|_{\mathrm{op}}=\lambda_{\min}(V_\bullet)^{-1}\le\lambda_n^{-1}(1+o(1))$ from the
hypothesis,
\[
\big\|V_\bullet^{-1}\widehat V_\bullet-I_v\big\|_{\mathrm{op}}
\le\lambda_{\min}(V_\bullet)^{-1}\,\|\widehat V_\bullet-V_\bullet\|_{\mathrm{op}}
\xrightarrow{p}0 ,
\]
which is the claim.
\end{proof}

\begin{lemma} \label{lemma:vadj_to_vtrue_scalar}
Suppose Assumptions \ref{asmp:psi_dependence} and \ref{asmp:vadj_to_vtrue} hold and $Y_{n,i}$ is scalar. 
Then,
\begin{equation}
\frac{1}{\lambda_{n}}\left(V_{true}-V_{adj}\right)\xrightarrow{p}0.
\end{equation}
\end{lemma}

\begin{proof}[Proof of \Cref{lemma:vadj_to_vtrue_scalar}]
Let $\vartheta_1, \vartheta_2, \vartheta_3$ denote arbitrary positive and finite constants.

\begin{align*}
V_{adj}-V_{true} & =\sum_{m=1}^{M}\sum_{t=1}^{T-m}\left(\omega\left(m,M\right)-1\right) \left( \Cov\left(y_{t},y_{t+m}\right)+\Cov\left(y_{t+m},y_{t}\right) \right)\\
 & \quad+\sum_{m\geq1}\sum_{t=1}^{T-m}\sum_{g=1}^{G}\sum_{i\in N_{t,g}^{T\cap G}}\sum_{j\in N_{t+m,g}^{T\cap G}}\left(\Cov\left(Y_{n,i},Y_{n,j}\right)+\Cov\left(Y_{n,j},Y_{n,i}\right)\right)\\
 & \quad+\sum_{m=M+1}^{T-1}\sum_{t=1}^{T-m}\Cov\left(y_{t},y_{t+m}\right)+\sum_{m=M+1}^{T-1}\sum_{t=1}^{T-m}\Cov\left(y_{t+m},y_{t}\right). \tag{\stepcounter{equation}\theequation}
\end{align*}

By applying Theorem A1 from KMS, for $d_{n}\left(i,j\right)=s\geq1$, we obtain $\left|\Cov\left(Y_{n,i},Y_{n,j}\right)\right|\leq\vartheta_{2}\theta_{n,s}^{1-\frac{2}{p}}$, 
where $\vartheta_{2}=C\left(\mu\vee1\right)^{2}\bar{\theta}$ for some constant $C>0$. 
Then,
\begin{align*}
 & \left|\sum_{m=1}^{M}\sum_{t=1}^{T-m}\left(\omega\left(m,M\right)-1\right)\Cov\left(y_{t},y_{t+m}\right)\right|\\
 & \leq\sum_{m=1}^{M}\sum_{t=1}^{T-m}\left|\omega\left(m,M\right)-1\right|\sum_{i\in N_{t}^{T}}\sum_{j\in N_{t+m}^{T}}\left|\Cov\left(Y_{n,i},Y_{n,j}\right)\right|\\
 & \leq\sum_{m=1}^{M}\left|\omega\left(m,M\right)-1\right|\sum_{t=1}^{T}\sum_{i\in N_{t}^{T}}\left|N_{t(i)+m}^{T}\right|\vartheta_{2}\theta_{n,m}^{1-\frac{2}{p}}\\
 & \leq\sum_{m=1}^{M}\left|\omega\left(m,M\right)-1\right|\sum_{i\in N_{n}}\left|N_{t(i)+m}^{T}\right|\vartheta_{2}\theta_{n,m}^{1-\frac{2}{p}}\\
 & \leq\vartheta_{2}n\sum_{m=1}^{M}\left|\omega\left(m,M\right)-1\right|\delta_{n}^{\partial}\left(m;1\right)\theta_{n,m}^{1-\frac{2}{p}}, \tag{\stepcounter{equation}\theequation}
\end{align*}
so it suffices for this scaled object to converge. 

Now, we show that $\frac{1}{\lambda_n} \sum_{m\geq1}\sum_{t=1}^{T-m}\sum_{g=1}^{G}\sum_{i\in N_{t,g}^{T\cap G}}\sum_{j\in N_{t+m,g}^{T\cap G}}\Cov\left(Y_{n,i},Y_{n,j}\right)$
is asymptotically negligible:
\begin{align*}
 & \sum_{m\geq1}\sum_{t=1}^{T-m}\sum_{g=1}^{G}\sum_{i\in N_{t,g}^{T\cap G}}\sum_{j\in N_{t+m,g}^{T\cap G}}\left|\Cov\left(Y_{n,i},Y_{n,j}\right)\right|\\
 & \leq\sum_{m\geq1}\sum_{t=1}^{T-m}\sum_{g=1}^{G}\left|N_{t,g}^{T\cap G}\right|\left|N_{t+m,g}^{T\cap G}\right|\vartheta_{2}\theta_{n,m}^{1-\frac{2}{p}}. \tag{\stepcounter{equation}\theequation}
\end{align*}
The object converges to zero after scaling due to the assumption used.

Finally, we show $\frac{1}{\lambda_n}\sum_{m=M+1}^{T-1}\sum_{t=1}^{T-m}\Cov\left(y_{t},y_{t+m}\right)$ is asymptotically negligible: 
\begin{align*}
\sum_{m=M+1}^{T-1}\sum_{t=1}^{T-m}\left|\Cov\left(y_{t},y_{t+m}\right)\right| & \leq\sum_{m=M+1}^{T-1}\sum_{t=1}^{T-m}\left|N_{t}^{T}\right|\left|N_{t+m}^{T}\right|\vartheta_{3}\theta_{n,m}^{1-\frac{2}{p}}. \tag{\stepcounter{equation}\theequation}
\end{align*}
The object converges to zero after scaling due to the assumption used.
\end{proof}

\subsection{Proof of Corollaries}

\begin{proof}[Proof of Corollary~\ref{cor:twoway_optimal}]
Write $P_G := J_G/G$, $Q_G := I_G - P_G$ and $P_T, Q_T$ analogously;
these are orthogonal idempotents on each factor. Substituting
$J_G = G P_G$, $J_T = T P_T$, $I = P + Q$ into $W_n^{2\text{-way}}$ and
collecting terms in the commuting blocks $\{P_T,Q_T\}\otimes\{P_G,Q_G\}$,
\[
W_n^{2\text{-way}}
= (G{+}T{-}1)\,P_T\!\otimes\!P_G
+ (G{-}1)\,Q_T\!\otimes\!P_G
+ (T{-}1)\,P_T\!\otimes\!Q_G
- Q_T\!\otimes\!Q_G .
\]
The coefficients are the eigenvalues; the only negative one is $-1$, on
$Q_T\otimes Q_G$ (nonempty iff $G,T\ge2$). The positive part zeros this
block, so $A_n^\star = W_n^{2\text{-way}} + Q_T\otimes Q_G$. Finally
$Q_T\otimes Q_G = (I_T - P_T)\otimes(I_G - P_G) = I_n - W_n^G/T - W_n^T/G + J_n/n = P_{00}$
by inclusion--exclusion, with the stated quadratic form.
\end{proof}

\begin{proof}[Proof of Corollary~\ref{cor:chs_optimal}]
Decompose $I_G = P_G + Q_G$ and $J_G - I_G = (G-1)P_G - Q_G$, and
substitute into $W_n^{\mathrm{CHS}}$. Since $P_G, Q_G$ are orthogonal
idempotent,
\[
W_n^{\mathrm{CHS}}
= \underbrace{\big[J_T + (G-1)C_M\big]}_{=:A}\otimes P_G
\;+\; \underbrace{\big[J_T - C_M\big]}_{=B}\otimes Q_G,
\]
an orthogonal direct sum over $\mathbb{R}^T\otimes\mathrm{range}(P_G)$ and
$\mathbb{R}^T\otimes\mathrm{range}(Q_G)$, so the positive part acts
blockwise: $(W_n^{\mathrm{CHS}})^+ = A^+\otimes P_G + B^+\otimes Q_G$.
Both $J_T$ and $C_M$ are PSD (the latter a Fej\'er kernel), so
$A = J_T + (G-1)C_M\succeq0$ and $A^+ = A$. For the second block,
$B^+ = B + B^-$ with $B^- = (C_M - J_T)_+\succeq0$, giving
$B^+\otimes Q_G = B\otimes Q_G + B^-\otimes Q_G$. Summing,
$A_n^\star = W_n^{\mathrm{CHS}} + B^-\otimes Q_G$. Since $Q_G$ acts on the
cluster index to produce $\check Y$, the quadratic form is
$Y_n'W_n^{\mathrm{CHS}}Y_n + \sum_g \check Y_{g,\cdot}'B^-\check Y_{g,\cdot}$.
\end{proof}

\begin{proof}[Proof of \Cref{cor:coverage}]
Write $\sigma_n^2:=\Var(S_n)=\nu'V_{\mathrm{true}}\nu$ and
$\phi_n:=\nu'V_{\mathrm{true}}\nu/\nu'V_{\mathrm{ET}}^\star\nu$; by
\Cref{thm:panel_dep_consistent_main}(a),(c),
$\nu'V_{\mathrm{ET}}^\star\nu\ge\nu'V_{\mathrm{adj}}\nu=\nu'V_{\mathrm{true}}\nu\,(1+o(1))$,
so $\phi_n\le1+o(1)$. By \Cref{thm:panel_dep_consistent_main}(b),
$\nu'\widehat V_{\mathrm{ET}}^\star\nu=\nu'V_{\mathrm{ET}}^\star\nu\,(1+o_p(1))$, hence
\[
r_n^2:=\frac{\sigma_n^2}{\nu'\widehat V_{\mathrm{ET}}^\star\nu}
=\phi_n\,\frac{\nu'V_{\mathrm{ET}}^\star\nu}{\nu'\widehat V_{\mathrm{ET}}^\star\nu}
=\phi_n\,(1+o_p(1))\ \le\ 1+o_p(1).
\]
Since $|S_n|/\sqrt{\nu'\widehat V_{\mathrm{ET}}^\star\nu}=(|S_n|/\sigma_n)\,r_n$ and
\Cref{thm:clt} gives $|S_n|/\sigma_n\xrightarrow{d}|Z|$, $Z\sim N(0,1)$, fix
$\varepsilon>0$. On $\{r_n\le1+\varepsilon\}$,
$\{(|S_n|/\sigma_n)\,r_n>z_{1-\alpha/2}\}\subseteq
\{|S_n|/\sigma_n>z_{1-\alpha/2}/(1+\varepsilon)\}$, so
\[
\Pr\!\Big\{\tfrac{|S_n|}{\sqrt{\nu'\widehat V_{\mathrm{ET}}^\star\nu}}>z_{1-\alpha/2}\Big\}
\ \le\ \Pr\!\Big\{\tfrac{|S_n|}{\sigma_n}>\tfrac{z_{1-\alpha/2}}{1+\varepsilon}\Big\}
+\Pr\{r_n>1+\varepsilon\}.
\]
The second term vanishes since $r_n\le1+o_p(1)$; the first, $z_{1-\alpha/2}/(1+\varepsilon)$
being a continuity point of the law of $|Z|$, converges to
$2\Phi\!\big(-z_{1-\alpha/2}/(1+\varepsilon)\big)$. Taking $\limsup_n$ and then
$\varepsilon\downarrow0$,
\[
\limsup_{n\to\infty}\Pr\!\Big\{\tfrac{|S_n|}{\sqrt{\nu'\widehat V_{\mathrm{ET}}^\star\nu}}>z_{1-\alpha/2}\Big\}
\ \le\ 2\Phi(-z_{1-\alpha/2})=\alpha .
\]

For the equality claim, $r_n^2=\phi_n(1+o_p(1))$ with $\phi_n\le1+o(1)$ deterministic, so
the limsup attains $\alpha$ iff $\phi_n\to1$; by (a) and (c) this holds iff
$\nu'(V_{\mathrm{ET}}^\star-V_{\mathrm{true}})\nu=o\!\big(\nu'V_{\mathrm{true}}\nu\big)$.
\end{proof}

\bibliographystyle{ecta}
\bibliography{het_means}

\begin{thebibliography}{13}
\newcommand{\enquote}[1]{``#1''}
\expandafter\ifx\csname natexlab\endcsname\relax\def\natexlab#1{#1}\fi

\bibitem[\protect\citeauthoryear{Abadie, Athey, Imbens, and Wooldridge}{Abadie et~al.}{2020}]{abadie2020sampling}
\textsc{Abadie, A., S.~Athey, G.~W. Imbens, and J.~M. Wooldridge} (2020): \enquote{Sampling-based versus design-based uncertainty in regression analysis,} \emph{Econometrica}, 88, 265--296.

\bibitem[\protect\citeauthoryear{Abadie, Athey, Imbens, and Wooldridge}{Abadie et~al.}{2023}]{abadie2023should}
---\hspace{-.1pt}---\hspace{-.1pt}--- (2023): \enquote{When should you adjust standard errors for clustering?} \emph{The Quarterly Journal of Economics}, 138, 1--35.

\bibitem[\protect\citeauthoryear{Andrews}{Andrews}{1991}]{andrews1991heteroskedasticity}
\textsc{Andrews, D.~W.} (1991): \enquote{Heteroskedasticity and autocorrelation consistent covariance matrix estimation,} \emph{Econometrica: Journal of the Econometric Society}, 817--858.

\bibitem[\protect\citeauthoryear{Cameron, Gelbach, and Miller}{Cameron et~al.}{2011}]{cameron2011robust}
\textsc{Cameron, A.~C., J.~B. Gelbach, and D.~L. Miller} (2011): \enquote{Robust inference with multiway clustering,} \emph{Journal of Business \& Economic Statistics}, 29, 238--249.

\bibitem[\protect\citeauthoryear{Casini}{Casini}{2023}]{casini2023theory}
\textsc{Casini, A.} (2023): \enquote{Theory of evolutionary spectra for heteroskedasticity and autocorrelation robust inference in possibly misspecified and nonstationary models,} \emph{Journal of Econometrics}, 235, 372--392.

\bibitem[\protect\citeauthoryear{Chan}{Chan}{2022}]{chan2022optimal}
\textsc{Chan, K.~W.} (2022): \enquote{Optimal difference-based variance estimators in time series: A general framework,} \emph{The Annals of Statistics}, 50, 1376--1400.

\bibitem[\protect\citeauthoryear{Chiang, Hansen, and Sasaki}{Chiang et~al.}{2024}]{chiang2024standard}
\textsc{Chiang, H.~D., B.~E. Hansen, and Y.~Sasaki} (2024): \enquote{Standard errors for two-way clustering with serially correlated time effects,} \emph{Review of Economics and Statistics}, 1--40.

\bibitem[\protect\citeauthoryear{Davezies, D’haultf{\oe}uille, and Guyonvarch}{Davezies et~al.}{2021}]{davezies2021empirical}
\textsc{Davezies, L., X.~D’haultf{\oe}uille, and Y.~Guyonvarch} (2021): \enquote{Empirical process results for exchangeable arrays,} \emph{The Annals of Statistics}, 49, 845--862.

\bibitem[\protect\citeauthoryear{Fama and French}{Fama and French}{1993}]{fama1993common}
\textsc{Fama, E.~F. and K.~R. French} (1993): \enquote{Common risk factors in the returns on stocks and bonds,} \emph{Journal of financial economics}, 33, 3--56.

\bibitem[\protect\citeauthoryear{Kojevnikov, Marmer, and Song}{Kojevnikov et~al.}{2021}]{kojevnikov2021limit}
\textsc{Kojevnikov, D., V.~Marmer, and K.~Song} (2021): \enquote{Limit theorems for network dependent random variables,} \emph{Journal of Econometrics}, 222, 882--908.

\bibitem[\protect\citeauthoryear{Newey and West}{Newey and West}{1987}]{newey1987asimple}
\textsc{Newey, W. and K.~West} (1987): \enquote{A Simple, Positive Semi-definite, Heteroskedasticity and Autocorrelation Consistent Covariance Matrix,} \emph{Econometrica}, 55, 703--708.

\bibitem[\protect\citeauthoryear{Xu and Yap}{Xu and Yap}{2024}]{xu2024clustering}
\textsc{Xu, R. and L.~Yap} (2024): \enquote{Clustering with Potential Multidimensionality: Inference and Practice,} \emph{arXiv preprint arXiv:2411.13372}.

\bibitem[\protect\citeauthoryear{Yap}{Yap}{2025}]{yap2025asymptotic}
\textsc{Yap, L.} (2025): \enquote{Asymptotic theory for two-way clustering,} \emph{Journal of Econometrics}, 249, 106001.

\end{thebibliography}

\end{document}